\documentclass[aps,pra,reprint,10pt,a4paper]{revtex4-1}
\usepackage[utf8]{inputenc}
\usepackage[sc,osf]{mathpazo}
\usepackage[T1]{fontenc}
\usepackage{amsfonts}
\usepackage{amssymb}
\usepackage{amsmath}
\usepackage{amsthm}
\usepackage{graphics}
\usepackage{graphicx}
\usepackage{braket}
\usepackage[colorlinks=true,linkcolor=blue,
urlcolor=blue,citecolor=blue]{hyperref}

\newtheorem{definition}{Definition}
\newtheorem{proposition}{Proposition}

\newtheorem{observation}{Observation}

\begin{document}
\title{Distinguishability classes, resource sharing, and bound entanglement distribution}

\author{Saronath Halder}
\thanks{Present affiliation: Harish-Chandra Research Institute, HBNI, Chhatnag Road, Jhunsi, Prayagraj (Allahabad) 211 019, India}
\email{saronath.halder@gmail.com}
\affiliation{Department of Mathematical Sciences, Indian Institute of Science Education and Research Berhampur, Transit Campus, Government ITI, Berhampur - 760010, Odisha, India}

\author{Ritabrata Sengupta}
\email{rb@iiserbpr.ac.in}
\affiliation{Department of Mathematical Sciences, Indian Institute of Science Education and Research Berhampur, Transit Campus, Government ITI, Berhampur - 760010, Odisha, India}

\begin{abstract}
Suppose a set of $m$-partite, $m\geq3$, pure orthogonal fully separable states is given. We consider the task of distinguishing these states perfectly by local operations and classical communication (LOCC) in different $k$-partitions, $1<k<m$. Based on this task, it is possible to classify the sets of product states into different classes. For tripartite systems, a classification of the sets with explicit examples is presented. Few important cases related to the aforesaid task are also studied when the number of parties, $m\geq4$. These cases never appear for a tripartite system. However, to distinguish any LOCC indistinguishable set, entanglement can be used as resource. An important objective of the present study is to learn about the efficient ways of resource sharing among the parties. We also find an interesting application of multipartite product states which are LOCC indistinguishable in a particular $k$-partition. Starting from such product states, we constitute a protocol to distribute bound entanglement between two spatially separated parties by sending a separable qubit. 
\end{abstract}
\maketitle

\section{Introduction}\label{sec1}
Quantum state discrimination problem is a well-known problem within the theory of quantum information. In this problem a quantum system is prepared in a state which is secretly chosen from a known set. The objective is to determine the state of the system. In brief, we say this as the problem of distinguishing a given set. If the states within a set are orthogonal to each other then by performing a measurement on the entire system, it is always possible to distinguish the set. On the other hand, perfect discrimination of nonorthogonal states is not possible \cite{Nielsen00}. Nevertheless, even if the given states are  orthogonal, the problem may arise if the parts of a composite system are distributed among spatially separated parties. This is because in a spatially separated configuration we consider that the parties can carry out measurements only on their own subsystems. But they can communicate with each other classically. This class of operations is known as local operations and classical communication (LOCC) \cite{Chitambar14}. 

In many research works \cite{Peres91, Walgate00, Virmani01, Ghosh01, Horodecki03, Fan04, Ghosh04, Nathanson05, Watrous05, Fan07, Duan07, Bandyopadhyay09, Yu12}, the authors considered the problem of distinguishing a given set when the parties are allowed to perform LOCC only. In brief, we say this as the problem of distinguishing a given set by LOCC. It is also known as the local state discrimination problem. The states we consider here are orthogonal to each other. Naturally, perfect discrimination of the states is studied. If a set is not perfectly distinguishable by LOCC then the set is an LOCC/locally indistinguishable set or simply a nonlocal set. If such a set forms a basis for any Hilbert space then the basis is an LOCC/locally indistinguishable basis or simply a nonlocal basis. The nonlocal sets can be used for practical purposes, for example, quantum secret sharing \cite{Markham08}, data hiding \cite{Terhal01, Eggeling02} etc. Furthermore, to study the nonlocal features of distributed quantum systems, different state discrimination problems by LOCC are crucial. 

If an orthonormal basis contains only entangled states then the basis must be locally indistinguishable, for example, the Bell basis \cite{Ghosh01}. However, there also exist nonlocal orthonormal bases which contain only product states. Clearly, such bases exhibit {\it quantum nonlocality without entanglement} as introduced for the very first time in Ref.~\cite{Bennett99-1}. Operationally, these bases correspond to separable measurements which cannot be accomplished by LOCC. Other than orthonormal product bases, there also exist small sets of pure orthogonal product states which are nonlocal \cite{Bennett99, DiVincenzo03}. Subsequently, many articles were written to study both bipartite and multipartite pure orthogonal product states \cite{Groisman01, Walgate02, Rinaldis04, Niset06, Ye07, Feng09, Duan10, Yang13, Childs13, Zhang14, Yu15, Zhang15, Wang15, Chen15, Yang15, Zhang16, Xu16, Zhang16-1, Xu16-1, Wang17, Zhang17, Xu17, Wang16, Zhang17-1, Zhang17-2, Zhang17-3, Corke17, Halder18, Halder18-1}. Nonetheless, the properties of multipartite product states across every bipartition have been studied very recently \cite{Halder19, Rout19}. In both the articles, sets of tripartite pure orthogonal product states have been constructed. Interestingly, these sets are locally indistinguishable across all bipartitions. Thus, for the completeness it is necessary to study the $m$-partite ($m\geq3$) product states which are locally distinguishable in some (or in all) $k$-partitions ($1<k<m$). Remember that for a multipartite system if a particular partition of the parties is mentioned then two or more parties can come together in a single location. Thereafter, the parties of a single location, can perform quantum operations together under the setting of LOCC.  

To distinguish a given nonlocal set, entanglement can be used as resource. There are several articles where entanglement in pure form was used to distinguish nonlocal sets \cite{Cohen07, Cohen08, Bandyopadhyay09-1, Bandyopadhyay10, Duan14, Bandyopadhyay14, Bandyopadhyay16, Zhang16-2, Bandyopadhyay18, Zhang18, Halder18, Li19}. In Ref.~\cite{Cohen08}, the author presented protocols to distinguish product states using entanglement as resource. After that article, such protocols for both bipartite \cite{Zhang16-2, Li19} and multipartite \cite{Zhang18, Halder18, Li19} systems were constructed. These protocols are very efficient as they consume less entanglement with respect to many other schemes. Note that to make a protocol entanglement consumption wise efficient, it is required to share resource state(s) in a suitable way among the parties. In particular, if a given set is locally indistinguishable in some (not all) $k$-partitions then there might be constraint on the ways of resource sharing. Thus, to reduce the entanglement consumption we have to learn about the efficient ways of resource sharing among the parties. This clearly indicates the importance of studying local (in)distinguishability of a given set across every $k$-partition.

Other than entanglement assisted state discrimination, there are many instances where entanglement is used as resource. For example, consider teleportation \cite{Bennett93}, superdense coding \cite{Bennett92}, cryptography \cite{Ekert91}, etc. Therefore, distribution of entanglement among spatially separated parties is important to use it as resource. Interestingly, it is possible to distribute entanglement between two spatially separated parties by sending a separable qubit. This was first introduced in Ref.~\cite{Cubitt03}. Subsequently, many related articles (both theoretical and experimental) were written \cite{Mista08, Mista09, Mista13, Kay12, Fedrizzi13, Vollmer13, Peuntinger13, Streltsov12, Chuan12, Streltsov14, Streltsov15, Zuppardo16}. To the best of our knowledge, in all the cases the distributed entangled state is negative under partial transpose in the desired bipartition. So, we ask the question how to distribute entangled state which is positive under partial transpose \cite{Horodecki97} by sending a separable qubit. Clearly, this implies the distribution of bound entanglement \cite{Horodecki98}. We mention here that the use of specific bound entangled states as resource was shown in many contexts. For example, consider the instances of secure key distillation \cite{Horodecki05, Horodecki08, Horodecki09}, quantum metrology \cite{Toth18}, quantum steering \cite{Moroder14}, quantum nonlocality \cite{Vertesi14, Yu17}, etc.

In this work we consider LOCC (in)distinguishability of a set of multipartite orthogonal product states in different $k$-partitions. To understand the present results, we give basic assumptions and the notations in Sec.~\ref{sec2}. Next, in Sec.~\ref{sec3}, different distinguishability classes of tripartite sets of product states are presented. This classification is based on the LOCC (in)distinguishability of the sets in different bipartitions. Along with the classification explicit examples are also given. After that few multipartite (number of parties, $m\geq4$) cases are also discussed in Sec.~\ref{sec4}. These cases do not appear for a tripartite system. It is also explained how LOCC (in)distinguishability of a set in different $k$-partitions is connected with the efficient ways of resource sharing. The resource states are shared among the parties in order to distinguish a given nonlocal set perfectly. As a consequence of the present study, we show how to distribute bound entanglement between two spatially separated parties by sending a separable qubit. The corresponding protocol is constituted in Sec.~\ref{sec5}, starting from a particular type of tripartite set of product states. Such a protocol is an indirect one to establish entanglement \cite{Zuppardo16}. Now, for an indirect protocol it is necessary to consider at least three particles as introduced in Ref.~\cite{Cubitt03} (see Ref.~\cite{Zuppardo16} for greater details). The present protocol also consists of only three particles. However, after the description of the protocol finally, in Sec.~\ref{sec6}, the conclusion is drawn.

\section{Assumptions and Notations}\label{sec2}
We consider here that the given states are pure multipartite (number of parties, $m\geq3$) fully separable states. Such a state has the form $\ket{\alpha_1}\otimes\ket{\alpha_2}\otimes\dots\otimes\ket{\alpha_m}\equiv\ket{\alpha_1}\ket{\alpha_2}\dots\ket{\alpha_m}$. These states are pairwise orthogonal to each other and thus, the perfect discrimination of them is considered. If the state $\ket{\alpha_1}\ket{\alpha_2}\dots\ket{\alpha_m}$ belongs to $\mathcal{H}=\mathcal{H}_1\otimes\mathcal{H}_2\otimes\dots\otimes\mathcal{H}_m$, then $\ket{\alpha_1}\in\mathcal{H}_1$, $\ket{\alpha_2}\in\mathcal{H}_2$ etc. Throughout the manuscript we use the notation $\ket{\alpha_1\pm \alpha_2 \pm\dots}\equiv\frac{1}{N}(\ket{\alpha_1}\pm\ket{\alpha_2}\pm\dots)$, $N$ is the normalization constant. For simplicity, we ignore the normalization constants in many places where these constants do not play any important role. Note that a given set may or may not be a complete basis in a Hilbert space and the states within a given set are equally probable. Unless, it is described clearly, consider that all the parties are spatially separated. Thus, the parties are restricted to perform measurements on their respective subsystems only. If there are many parties in a single location then those parties can perform measurements together under the setting of LOCC. For a tripartite system, there are only bipartitions. So, when a bipartition is considered, two parties come together in a single location. But if the number of parties, $m\geq4$ then there are $k$-partitions ($1<k<m$). Therefore, the number of parties which can come together, may depend on the partition. We also mention that in case of entanglement-assisted discrimination of a given nonlocal set, we consider only pure entangled states as resource. These entangled states can be bipartite or multipartite depending on the situation. We now proceed to analyze different examples of tripartite systems.  

\section{Tripartite Systems}\label{sec3}
Consider three parties A(lice), B(ob), and C(harlie). A tripartite quantum system is associated with a Hilbert space, $\mathcal{H} = \mathcal{H}_A\otimes\mathcal{H}_B\otimes\mathcal{H}_C = \mathbb{C}^{d_A}\otimes\mathbb{C}^{d_B}\otimes\mathbb{C}^{d_C}$; $d_A$, $d_B$, and $d_C$ are the dimensions of the subsystems. These subsystems are possessed by Alice, Bob, and Charlie respectively. In this section we use the notation $i|jk$ which stands for a bipartition: $i$ versus $jk$ ($j$ and $k$ are in a single location). For tripartite systems, there are three bipartitions - $A|BC$, $B|CA$, and $C|AB$. Based on these bipartitions, we classify the sets of tripartite orthogonal product states into different categories. We say these as distinguishability classes, given as the following. 

{\bf Classification.}~(i) Locally indistinguishable across every bipartition, that is, such a set is a genuinely nonlocal set, (ii) locally indistinguishable across only two bipartitions, (iii) locally indistinguishable across only one bipartition, (iv) distinguishable across all of the bipartitions. Examples of tripartite genuinely nonlocal sets of product states can be found in Ref.~\cite{Halder19, Rout19}. In particular, in Ref.~\cite{Rout19}, a classification of genuinely nonlocal product bases and their entanglement assisted discrimination were discussed. However, in this work we focus on those sets which are locally distinguishable in at least one bipartition. If a set is locally distinguishable in at least one bipartition then that set is a {\it bi-distinguishable} set. Moreover, a {\it fully bi-distinguishable} set is distinguishable across every bipartition but it is locally indistinguishable when all the parties are spatially separated. Again, a fully distinguishable (or simply distinguishable) set is locally distinguishable when all the parties are spatially separated. A computational basis in any given Hilbert space is an example of a fully distinguishable set of product states. 

Now, before we present the tripartite results, it is important to mention here about certain features of the product bases presented in Ref.~\cite{Bennett99-1}. Two separate product bases are constructed in that paper. One is in $\mathbb{C}^3\otimes\mathbb{C}^3$ and the other is in $\mathbb{C}^2\otimes\mathbb{C}^2\otimes\mathbb{C}^2$. The two-qutrit product basis consists of nine states: \{$\ket{0}\ket{0\pm1}$, $\ket{0\pm1}\ket{2}$, $\ket{2}\ket{1\pm2}$, $\ket{1\pm2}\ket{0}$, $\ket{1}\ket{1}$\}. These states cannot be perfectly distinguished by LOCC. However, to exhibit local indistinguishability, first eight states among the nine are sufficient. The three-qubit product basis consists of eight states: \{$\ket{0}\ket{1}\ket{0\pm1}$, $\ket{1}\ket{0\pm1}\ket{0}$, $\ket{0\pm1}\ket{0}\ket{1}$, $\ket{0}\ket{0}\ket{0}$, $\ket{1}\ket{1}\ket{1}$\}. This set is also locally indistinguishable when the three qubits are distributed among three spatially separated parties. We note that $\mathbb{C}^3\otimes\mathbb{C}^3$ is the minimum dimensional bipartite Hilbert space configuration for such product states to exist. This is because of the following proposition,  discussed previously in Refs.~\cite{Bennett99, DiVincenzo03}.
\begin{proposition}\label{prop1}
Any set of orthogonal product states in $\mathbb{C}^2\otimes\mathbb{C}^d$ can be perfectly distinguished by LOCC and therefore, the set can be extended to a complete basis.
\end{proposition}
Due to the above proposition, any three-qubit pure orthogonal product states can be perfectly distinguished by LOCC across every bipartition. Therefore, the three-qubit product basis constructed in Ref.~\cite{Bennett99-1} is locally distinguishable across every bipartition. But the set can show local indistinguishability when all the parties are spatially separated (as mentioned earlier). This constitutes an example of a fully bi-distinguishable set. Nevertheless, in higher dimensions (when the dimension of each subsystem, $d\geq3$) Proposition \ref{prop1} does not work. So, it is important to find fully bi-distinguishable sets in higher dimensions. We identify that a particular type of tripartite sets in $(\mathbb{C}^d)^{\otimes3}$, $d\geq3$, given in Ref.~\cite{Halder18} are fully bi-distinguishable. The states of the sets are 
\begin{equation}\label{eq1}
\begin{array}{c}
|\psi_{1i}\rangle=|0\rangle|i\rangle|0+i\rangle,~
|\psi_{1i}^\perp\rangle=|0\rangle|i\rangle|0-i\rangle, \\ 
|\psi_{2i}\rangle=|i\rangle|0+i\rangle|0\rangle,~
|\psi_{2i}^\perp\rangle=|i\rangle|0-i\rangle|0\rangle, \\
|\psi_{3i}\rangle=|0+i\rangle|0\rangle|i\rangle,~
|\psi_{3i}^\perp\rangle=|0-i\rangle|0\rangle|i\rangle, 
\end{array}
\end{equation}
where $i=1,2,\dots,(d-1)$. The proof of local indistinguishability of the above sets when all the parties are spatially separated can be found in Ref.~\cite{Halder18}. Here we prove that such a set is locally distinguishable across every bipartition by constructing an explicit protocol. 

{\bf Protocol 1.}~Consider any two parties together: $AB$, $BC$, or $CA$ (in the cyclic order). These two parties can perform a measurement, described by the projectors: $\mathcal{P}_i$ =$\ket{0i}\bra{0i}$, $\mathcal{P}_f$ = $\mathcal{I}-\sum_i\mathcal{P}_i$; $i=1,2,\dots,(d-1)$. Corresponding to each outcome ``$i$'', there are two orthogonal pure states left, which can be perfectly distinguished by LOCC \cite{Walgate00}. For the outcome ``$f$'', the other party (who stands alone) performs a measurement in computational basis. Here also corresponding to each outcome, there are two orthogonal pure states left, which can be perfectly distinguished by LOCC \cite{Walgate00}. Thus, the protocol completes and it is proved that the above sets are the examples of fully bi-distinguishable sets. 

In Ref.~\cite{Halder18}, it was already shown that a two-qubit maximally entangled state is sufficient to distinguish a set of Eq.~(\ref{eq1}). Interestingly, the resource state can be shared between any two parties. This is because of the following facts: (a) For a fixed $d$, the given set and the resource state together generate a new set in higher dimension. (b) The resource state can be shared between any two parties. But each time the newly generated set captures the same structure when all three parties are spatially separated. Thus, the local discrimination of a newly generated set implies that the resource state can be shared between any two parties. However, if a given tripartite set is not fully bi-distinguishable, then which pair of parties share the resource may play a vital role. Subsequently, we show that a bipartite entangled state can be sufficient to distinguish a set but it must not be shared between particular pair(s) of parties. Now for any fully bi-distinguishable tripartite set the following proposition is presented. 
\begin{proposition}\label{prop2}
If a fully bi-distinguishable set in $\mathbb{C}^{d_A}\otimes\mathbb{C}^{d_B}\otimes\mathbb{C}^{d_C}$, $d_A\geq d_B\geq d_C$, is given then a maximally entangled bipartite state in $\mathbb{C}^{d}\otimes\mathbb{C}^{d}$, $d=d_B$, shared between any two parties is sufficient to distinguish the set locally.
\end{proposition}
\begin{proof} 
Consider that a fully bi-distinguishable set in $\mathbb{C}^{d_A}\otimes\mathbb{C}^{d_B}\otimes\mathbb{C}^{d_C}$, $d_A\geq d_B\geq d_C$, is given. According to the definition, such a set is locally distinguishable across all the bipartitions. Also consider that a maximally entangled bipartite state $\ket{\Psi}$ in $\mathbb{C}^{d}\otimes\mathbb{C}^{d}$, $d=d_B$ is available as resource to distinguish the given set. If $\ket{\Psi}$ is shared between Alice and Bob then using the resource state Bob can teleport his subsystem to Alice. This is equivalent to the situation where Alice and Bob are together in a single location. In this way, they constitute the bipartition $AB|C$. Similarly, if $\ket{\Psi}$ is shared between Alice and Charlie then using it Charlie can teleport his subsystem constituting the bipartition $AC|B$. Finally, if $\ket{\Psi}$ is shared between Bob and Charlie then any of them can teleport his subsystem to the location of the other party constituting the bipartition $A|BC$. After constituting any of the bipartitions the set can be perfectly distinguished by LOCC.
\end{proof}
We now present the examples of tripartite bi-distinguishable sets. It is easy to construct a tripartite set which is locally distinguishable in only one bipartition. By Proposition \ref{prop1}, one can infer that to construct such a set, at least two of the subsystems must have dimensions $\geq3$. But the dimension of the other subsystem can be $2$. Now, consider a basis $\{\ket{\psi_i}_{AB}\ket{j}_{C}\}$ in $\mathbb{C}^3\otimes\mathbb{C}^3\otimes\mathbb{C}^2$; $j=0,1$. Here, $\ket{\psi_i}_{AB}$ are the bipartite locally indistinguishable product states in $\mathbb{C}^3\otimes\mathbb{C}^3$, constructed in Ref.~\cite{Bennett99-1}, $\forall i =0,1,\dots,8$. The tripartite set $\{\ket{\psi_i}_{AB}\ket{j}_{C}\}$ is locally distinguishable in only one bipartition, that is, $AB|C$ bipartition. Interestingly, to distinguish this set, a bipartite entangled state can be sufficient. But this state must not be shared between Alice and Charlie or between Bob and Charlie. Notice that the local indistinguishability of this set is solely because of the states $\ket{\psi_i}_{AB}$, placed between Alice and Bob. Therefore, the resource state must be shared between Alice and Bob. From the construction, it is clear that this technique is not applicable to produce a tripartite set which is locally indistinguishable in a particular bipartition. So, the structure of this set must be nontrivial. For such a set to exist, dimension of one of the subsystems must be $\geq3$, while the dimensions of the other two subsystems can be $2$. We now construct a tripartite set of product states in $\mathbb{C}^3\otimes\mathbb{C}^2\otimes\mathbb{C}^2$ which is locally indistinguishable in a particular bipartition. The states of the set are
\begin{equation}\label{eq2}
\begin{array}{ll}
|\psi_{1}\rangle = |0-1\rangle|0\rangle|0\rangle, &
|\psi_{2}\rangle = |0+1\rangle|0\rangle|0\rangle, \\

|\psi_{3}\rangle = |2-0\rangle|1\rangle|0\rangle, &
|\psi_{4}\rangle = |2+0\rangle|1\rangle|0\rangle, \\

|\psi_{5}\rangle = |1\rangle|1\rangle|0-1\rangle, &
|\psi_{6}\rangle = |1\rangle|1\rangle|0+1\rangle, \\

|\psi_{7}\rangle = |1-2\rangle|0\rangle|1\rangle, &
|\psi_{8}\rangle = |1+2\rangle|0\rangle|1\rangle, \\

|\psi_{9}\rangle = |0\rangle|0-1\rangle|1\rangle, &
|\psi_{10}\rangle = |0\rangle|0+1\rangle|1\rangle.
\end{array}
\end{equation}
Notice that the dimension for which the above set is constructed is the minimum one for such a set to exist. By Proposition \ref{prop1}, the above set is locally distinguishable in $B|CA$ and in $C|AB$ bipartitions. Here we prove that the above set is locally indistinguishable in $A|BC$ bipartition. For that purpose, we consider the above set in $\mathbb{C}^3\otimes\mathbb{C}^4$ and relabel the $4$-dimensional vectors as $\ket{10}\rightarrow\ket{0}$, $\ket{11}\rightarrow\ket{1}$, $\ket{01}\rightarrow\ket{2}$, $\ket{00}\rightarrow\ket{3}$. After relabeling one can easily check that the states $\{\ket{\psi_i}\}_{i=3}^{10}$ turn into the states: $\{\ket{2\pm0}\ket{0}, \ket{1}\ket{0\pm1}, \ket{1\pm2}\ket{2}, \ket{0}\ket{1\pm2}\}$. These states are similar as the locally indistinguishable states in $\mathbb{C}^3\otimes\mathbb{C}^3$, constructed by Bennett {\it et al.} in Ref.~\cite{Bennett99-1}. The locally indistinguishable states of Bennett {\it et al.} are \{$\ket{0}\ket{0\pm1}$, $\ket{0\pm1}\ket{2}$, $\ket{2}\ket{1\pm2}$, $\ket{1\pm2}\ket{0}$\}. These eight states form a subset, $\mathcal{S}$ of the original two-qutrit basis as given earlier. Local indistinguishability of this subset was also explained in the same article. Now, if a subset of a set is locally indistinguishable then the original set must be. In this way, the above set of Eq.~(\ref{eq2}) is indistinguishable in $A|BC$ bipartition. Notice that a two-qubit maximally entangled state is sufficient to distinguish the above set. Obviously, the resource state must not be shared between Bob and Charlie. The resource state can be shared either between Alice and Bob or between Alice and Charlie. The classification of tripartite sets is also explained in Fig.~\ref{fig1}. Note that the diagram does not represent anything other than the classification. Recall that for a tripartite system there are only bipartitions. When the number of parties, $m\geq4$, then there are $k$-partitions. In particular when the number of parties, $m\geq4$, we can observe certain features which may not appear in a tripartite system. For this reason, in the following section we discuss few multipartite cases when the number of parties strictly larger than $3$.

\begin{figure}
\includegraphics[width=0.51\textwidth, height=0.31\textheight]{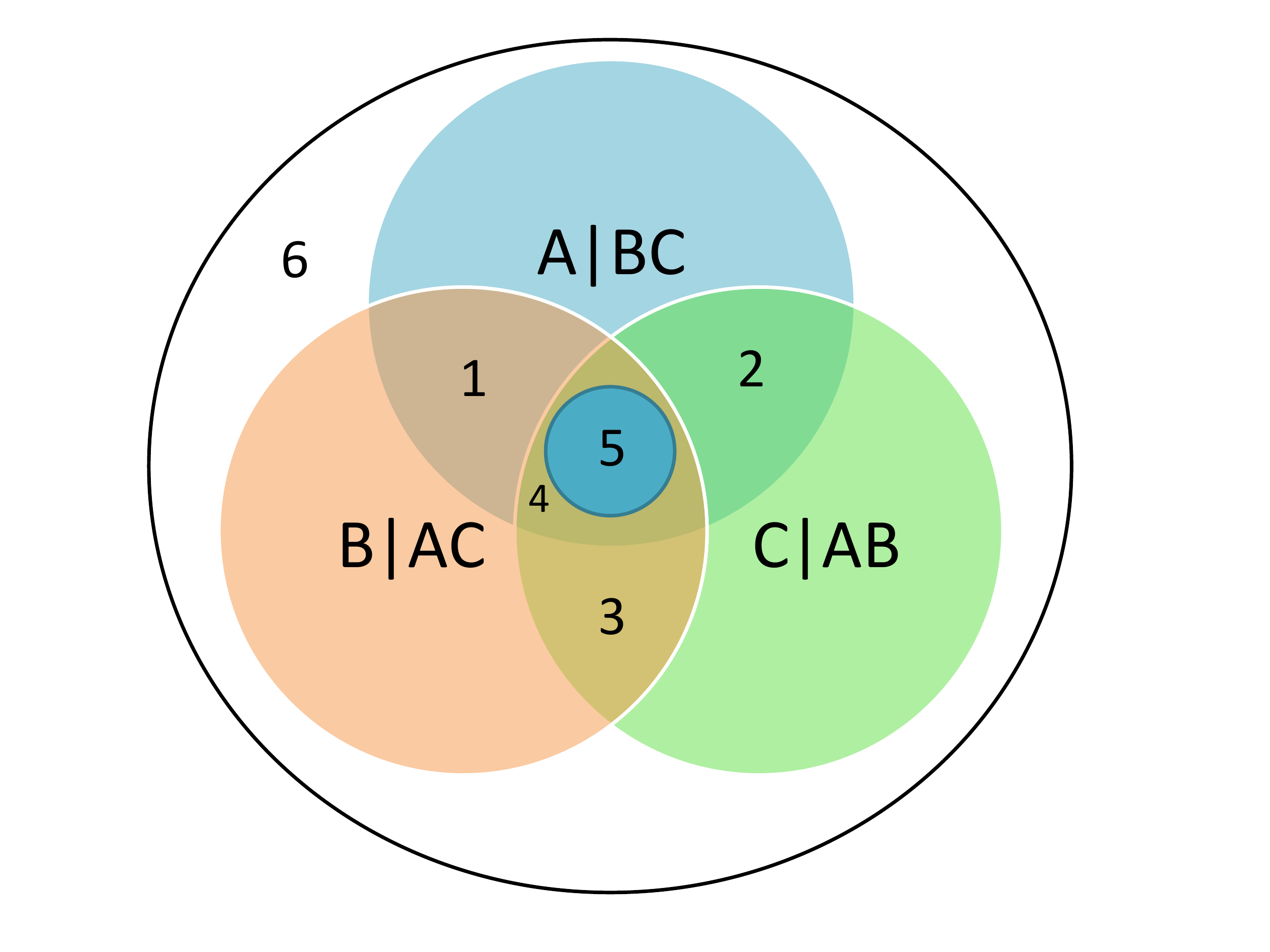}
\caption{(Color online) Classification of tripartite sets of product states: The largest circle contains all the sets. The medium sized circles contain sets that are locally distinguishable in at least one bipartition. The regions $1$, $2$, and $3$ contain sets which are locally distinguishable in only two bipartitions. The region $4$ contains sets which are locally distinguishable in all bipartitions. Fully distinguishable sets are contained in the region $5$ and the region $6$ contains the genuinely nonlocal sets.}\label{fig1}
\end{figure}

\section{Multipartite Systems}\label{sec4}
For an $m$-partite system, we label the parties as $A_1, A_2,\dots, A_m$. The corresponding Hilbert space is $\mathcal{H} = \mathcal{H}_{1}\otimes\mathcal{H}_{2}\otimes\dots\otimes\mathcal{H}_{m} = \mathbb{C}^{d_{1}}\otimes\mathbb{C}^{d_{2}}\otimes\dots\otimes\mathbb{C}^{d_{m}}$; $d_{1}$, $d_{2}$, $\dots$, $d_{m}$ are the dimensions of the subsystems. These subsystems are possessed by $A_1$, $A_2$, $\dots$, $A_m$ respectively. To realize a $k$-partition, we use the notation $A_{a_1}A_{a_2}\dots|A_{b_1}A_{b_2}\dots|A_{c_1}A_{c_2}\dots|\dots$, $a_i,b_i,c_i$ are the labels of $m$ parties. So, ``$|$'' stands for a partition and for a $k$-partition, there are $(k-1)$ ``$|$''s. We now start with the generalization of the sets  of Eq.~(\ref{eq1}) when $\mathcal{H} = (\mathbb{C}^{d})^{\otimes m}$, $m\geq4$. The states of those generalized sets are
\begin{equation}\label{eq3}
\begin{array}{c}
|\psi_{1i}\rangle=|0\rangle|0\rangle\cdots|0\rangle|i\rangle|0+i\rangle, \\
|\psi_{1i}^\perp\rangle=|0\rangle|0\rangle\cdots|0\rangle|i\rangle|0-i\rangle, \\[1ex]
|\psi_{2i}\rangle=|0\rangle\cdots|0\rangle|i\rangle|0+i\rangle|0\rangle, \\
|\psi_{2i}^\perp\rangle=|0\rangle\cdots|0\rangle|i\rangle|0-i\rangle|0\rangle, \\[1ex]
\vdots \\[1ex]
|\psi_{mi}\rangle=|0+i\rangle|0\rangle|0\rangle\cdots|0\rangle|i\rangle, \\
|\psi_{mi}^\perp\rangle=|0-i\rangle|0\rangle|0\rangle\cdots|0\rangle|i\rangle, 
\end{array}
\end{equation}
where $i$ = $1,2,\dots,(d-1)$. We mention here that this generalization for $m$-qubit case is originally given in Ref.~\cite{Xu16}, while the $m$-qudit case is given in Ref.~\cite{Halder18}. We first show that though the tripartite sets are fully bi-distinguishable but when the number of parties increases, the set may not be fully bi-distinguishable. Remember that because of the construction, if any two adjacent parties come together in single location then the sets become distinguishable. This can be easily shown following the same technique as given in Protocol 1 for any number of parties and for any dimension of the subsystems. Therefore, the question arises what happens when the parties coming together are not adjacent. If the number of parties, $m=4$, then there are $A_1, A_2, A_3, A_4$. The adjacent pairs of parties are $A_1A_2$, $A_2A_3$, $A_3A_4$, and $A_4A_1$ (considering cyclic order). The bipartitions are given by $A_1|A_2A_3A_4$, $A_2|A_3A_4A_1$, $A_3|A_4A_1A_2$, $A_4|A_1A_2A_3$, $A_1A_2|A_3A_4$, $A_4A_1|A_2A_3$, and $A_1A_3|A_2A_4$. Notice that except the bipartition $A_1A_3|A_2A_4$, in case of all other bipartitions, the adjacent parties come together. We now write the set of Eq.~(\ref{eq3}) for $m=4$ in $A_1A_3|A_2A_4$ bipartition and it is given by 
\begin{equation}\label{eq4}
\begin{array}{c}
\ket{0i}\ket{00\pm0i},~
\ket{00\pm0i}\ket{i0},\\
\ket{i0}\ket{00\pm i0},~
\ket{00\pm i0}\ket{0i}.\\
\end{array}
\end{equation}
As defined before $i=1,2,\dots,(d-1)$ and the Hilbert space is given by $(\mathbb{C}^d)^{\otimes4}$. For a fixed $i$, we relabel the states as: $\ket{0i}\rightarrow\ket{0}$, $\ket{00}\rightarrow\ket{1}$, and $\ket{i0}\rightarrow\ket{2}$. After relabeling, the above set turns into the eight indistinguishable states of $\mathcal{S}$ by Bennett {\it et al.} $\mathcal{S}$ is already mentioned in the previous section. Thus, the above set is indistinguishable only in $A_1A_3|A_2A_4$ bipartition. As a consequence, if only $A_1$ and $A_3$ come together and the other two parties stand alone then also the above set shows local indistinguishability. In this way, the above set is indistinguishable in certain tripartitions, for example, $A_1A_3|A_2|A_4$, $A_1|A_2A_4|A_3$. In general, if one considers $m$-partite ($m$ is even) generalization of the set of Eq.~(\ref{eq1}) then it is indistinguishable in $A_1A_3\dots A_{m-1}|A_2A_4\dots A_m$ bipartition. This follows from the fact that for a fixed $i$, a suitable relabeling shows that the set turns into a locally indistinguishable set, constructed in Ref.~\cite{Zhang15}. This bipartition also ensures that there are $k$-partitions ($2<k<m$) in which the set shows local indistinguishability. From this discussion we conclude that if the adjacent parties come together then only it helps in the distinguishability of the sets of Eq.~(\ref{eq3}) by LOCC. We now consider the same set for $m=5$. It can be easily shown that this set is distinguishable across every bipartition. Because in case of every bipartition at least a pair of adjacent parties come together. Thus, for $m=5$, the set of Eq.~(\ref{eq3}) forms a fully bi-distinguishable set. But there are tripartitions in which this set can show local indistinguishability, for example, $A_1A_3|A_2A_4|A_5$ tripartition. Because in this tripartition no adjacent parties come together. This leads us to the following observation.
\begin{observation}\label{obs1}
A fully bi-distinguishable set can show local indistinguishability across a $k$-partition, where $2<k<m$.
\end{observation}
To distinguish this generalized set, a bipartite entangled state can be sufficient but that state must be shared between any two adjacent parties. In general, if an $m$-partite nonlocal set is locally distinguishable across all $k$-partitions (for $1<k<m$) then the set is {\it fully $k$-distinguishable} set. To distinguish such a set perfectly by LOCC there exists a bipartite resource state which can be shared between any two parties. It will be interesting if it is possible to construct these sets for higher $m$ ($m\geq4$). So far, we have discussed many examples of multipartite nonlocal sets for which a bipartite entangled state can be sufficient to distinguish. On the other hand, there are multipartite nonlocal sets to distinguish which entangled resources across every bipartition are necessary \cite{Halder19, Zhang17, Rout19}. Therefore, the intermediate task is to construct an $m$-partite, $m\geq4$, nonlocal set to distinguish which an $m^\prime$-partite, $2<m^\prime<m$, entangled resource is necessary. Obviously, this question is related to a particular class of $m$-partite nonlocal sets. For which if any $m^\prime$ parties come together then the sets become distinguishable. However, Observation \ref{obs1} indicates that it may possible to construct an $m$-partite set, $m\geq4$, which is indistinguishable across every $k$-partition, $2<k<m$. But the set may show local distinguishability in certain bipartitions. Here, it is important to mention that if a set is locally indistinguishable across every bipartition then the set must be indistinguishable across every $k$-partition. We now construct a set which forms a complete basis in $(\mathbb{C}^3)^{\otimes4}$. The states of the basis are
\begin{widetext}
\begin{equation}\label{eq5}
\begin{array}{c}
\ket{0}\ket{0}\ket{1}\ket{0\pm1},~
\ket{0}\ket{0}\ket{2}\ket{0\pm2},~
\ket{2}\ket{1}\ket{0}\ket{0\pm1},~
\ket{1}\ket{1}\ket{2}\ket{0\pm1},~
\ket{2}\ket{1}\ket{2}\ket{0\pm2},\\

\ket{0}\ket{1}\ket{0\pm1}\ket{0},~
\ket{0}\ket{2}\ket{0\pm2}\ket{0},~
\ket{1}\ket{0}\ket{0\pm1}\ket{2},~
\ket{1}\ket{2}\ket{0\pm1}\ket{1},~
\ket{1}\ket{2}\ket{0\pm2}\ket{2},\\

\ket{1}\ket{0\pm1}\ket{0}\ket{0},~
\ket{2}\ket{0\pm2}\ket{0}\ket{0},~
\ket{0}\ket{0\pm1}\ket{2}\ket{1},~
\ket{2}\ket{0\pm1}\ket{1}\ket{1},~
\ket{2}\ket{0\pm2}\ket{2}\ket{1},\\

\ket{0\pm1}\ket{0}\ket{0}\ket{1},~
\ket{0\pm2}\ket{0}\ket{0}\ket{2},~
\ket{0\pm1}\ket{2}\ket{1}\ket{0},~
\ket{0\pm1}\ket{1}\ket{1}\ket{2},~
\ket{0\pm2}\ket{2}\ket{1}\ket{2},\\[1 ex]

\ket{0}\ket{0}\ket{0}\ket{0},~

\ket{0}\ket{0}\ket{1}\ket{2},~


\ket{0}\ket{1}\ket{0}\ket{1},~
\ket{0}\ket{1}\ket{0}\ket{2},~

\ket{0}\ket{1}\ket{1}\ket{1},~

\ket{0}\ket{1}\ket{2}\ket{0},\\
\ket{0}\ket{1}\ket{2}\ket{2},~

\ket{0}\ket{2}\ket{0}\ket{1},~
\ket{0}\ket{2}\ket{0}\ket{2},~

\ket{0}\ket{2}\ket{1}\ket{1},~

\ket{0}\ket{2}\ket{2}\ket{1},~
\ket{0}\ket{2}\ket{2}\ket{2},\\


\ket{1}\ket{0}\ket{1}\ket{0},~
\ket{1}\ket{0}\ket{1}\ket{1},~

\ket{1}\ket{0}\ket{2}\ket{0},~
\ket{1}\ket{0}\ket{2}\ket{1},~
\ket{1}\ket{0}\ket{2}\ket{2},~

\ket{1}\ket{1}\ket{0}\ket{1},\\
\ket{1}\ket{1}\ket{0}\ket{2},~

\ket{1}\ket{1}\ket{1}\ket{0},~
\ket{1}\ket{1}\ket{1}\ket{1},~

\ket{1}\ket{1}\ket{2}\ket{2},~

\ket{1}\ket{2}\ket{0}\ket{0},~

\ket{1}\ket{2}\ket{1}\ket{2},\\

\ket{1}\ket{2}\ket{2}\ket{0},~
\ket{1}\ket{2}\ket{2}\ket{1},~

\ket{2}\ket{0}\ket{0}\ket{1},~

\ket{2}\ket{0}\ket{1}\ket{0},~
\ket{2}\ket{0}\ket{1}\ket{2},~

\ket{2}\ket{0}\ket{2}\ket{0},\\
\ket{2}\ket{0}\ket{2}\ket{2},~

\ket{2}\ket{1}\ket{0}\ket{2},~

\ket{2}\ket{1}\ket{1}\ket{0},~
\ket{2}\ket{1}\ket{1}\ket{2},~

\ket{2}\ket{1}\ket{2}\ket{1},~

\ket{2}\ket{2}\ket{0}\ket{1},\\
\ket{2}\ket{2}\ket{0}\ket{2},~

\ket{2}\ket{2}\ket{1}\ket{0},~
\ket{2}\ket{2}\ket{1}\ket{1},~

\ket{2}\ket{2}\ket{2}\ket{0},~
\ket{2}\ket{2}\ket{2}\ket{2}.
\end{array}
\end{equation}
\end{widetext}
Notice that in the above basis first few states are twisted and other states are simple product states. We now present the following proposition for the above basis.
\begin{proposition}\label{prop3}
If any two parties come together then also the above basis shows local indistinguishability. Therefore, the basis must be indistinguishable across all tripartitions.
\end{proposition}
The proof is given in the Appendix \ref{appA}. However, it is easy to check that if three parties come together then the basis can be perfectly distinguished by LOCC. Furthermore, because of the construction if a party stands alone then that party cannot eliminate any state from the basis via orthogonality-preserving LOCC \cite{Halder18}. These facts imply that a resource state distributed among three spatially separated parties is necessary to distinguish the basis. We now present an observation. This is due to the resource related discussions given so far. 
\begin{observation}\label{obs2}
Given an $m$-partite nonlocal set, if to distinguish the set perfectly by LOCC an $m^\prime$-partite $(2\leq m^\prime \leq m)$ resource is necessary then the number $m^\prime$ must be minimized over all $k$-partitions, $1<k<m$.
\end{observation}
The above observation clearly indicates a {\it threshold problem} of distinguishability. Suppose a set of $m$-partite orthogonal product states is given. One has to find out the minimum number of parties which have to come together in a single location to distinguish the set perfectly. Note that the multipartite sets which are indistinguishable in a particular bipartition, may have application to distribute bound entanglement. A particular such scenario regarding bound entanglement distribution by sending a separable qubit is presented in the following section.

\section{Bound entanglement distribution}\label{sec5}
Here we discuss the problem of distributing bound entanglement by sending a separable qubit. The protocol we describe here consists of minimum number of particles. A schematic diagram is given in Fig.~\ref{fig2}. 
\begin{figure}[h!]
\includegraphics[width=0.40\textwidth, height=0.205\textheight]{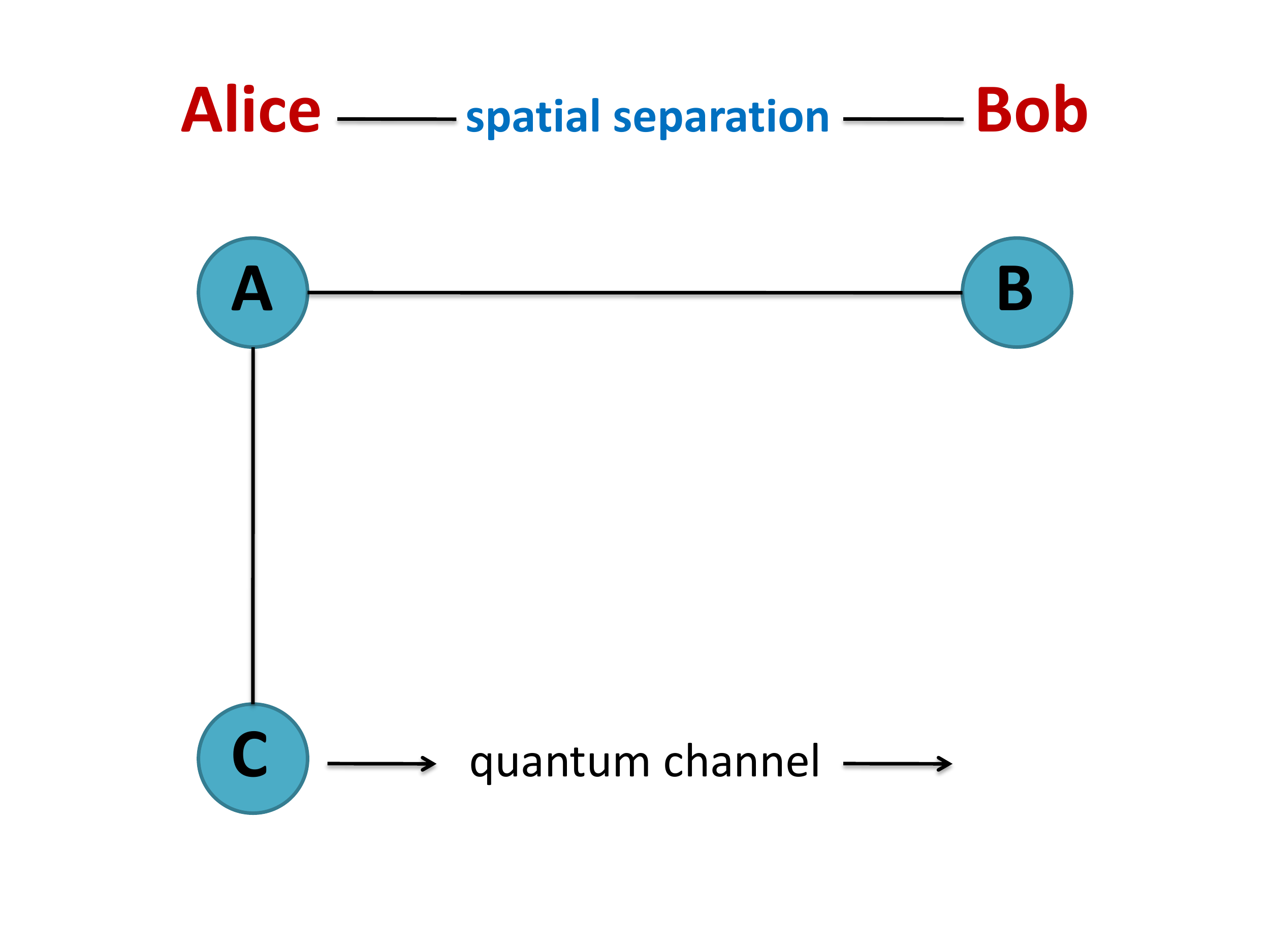}
\caption{(Color online) Initially Alice possesses both the particles A and C. Bob possesses the particle B. Initially there is no entanglement between Alice and Bob. Thereafter, Alice communicates the particle C to Bob via a noiseless quantum channel. Communicated entanglement is also zero. Finally, after the communication of C, there is bound entanglement between Alice and Bob.}\label{fig2}
\end{figure}
We start with a set of product states. To construct the set first consider certain states given in Eq.~(\ref{eq2}). These states are $\{\ket{\psi_1}, \ket{\psi_3}, \ket{\psi_5}, \ket{\psi_7}, \ket{\psi_9}\}$. Along with these states we add another state $\ket{s}=\ket{0+1+2}\ket{0+1}\ket{0+1}$. By Proposition \ref{prop1}, the new set must be locally distinguishable across $B|CA$ and $C|AB$ bipartitions (these symbols contain usual meaning as defined in Sec.~\ref{sec3}). Next, we define a tripartite mixed state in $\mathbb{C}^3\otimes\mathbb{C}^2\otimes\mathbb{C}^2$ and it is given by
\begin{equation}\label{eq6}
\rho_{ABC} = \frac{1}{6}(\mathbb{I}-\sum_{i=1}^6\ket{\phi_i}\bra{\phi_i}).
\end{equation}
Here the states $\{\ket{\phi_i}\}_{i=1}^6$ are the normalized version of the states $\{\ket{\psi_1}, \ket{\psi_3},\ket{\psi_5},\ket{\psi_7},\ket{\psi_9},\ket{s}\}$ respectively. $\mathbb{I}$ is the identity operator acting on the corresponding Hilbert space. Initially, this state is shared between Alice and Bob in a way that Alice possesses the particles A and C, and Bob possesses the particle B. Then Alice sends the particle C to Bob to establish bound entanglement between them. We say the initial entanglement as $E_{in}$ which is equal to $E_{AC|B}$. This is the amount of entanglement contained in the state $\rho_{ABC}$ in $AC|B$ bipartition. Similarly, the amount of entanglement which is communicated is $E_{com}$ = $E_{C|AB}$. This is the amount of entanglement contained in the state $\rho_{ABC}$ in $C|AB$ bipartition. Finally, after the communication of the particle, the entanglement is $E_{fin}$ = $E_{A|BC}$. This is the amount of entanglement contained by the state $\rho_{ABC}$ in $A|BC$ bipartition. Here we show that the state $\rho_{ABC}$ is separable in $B|AC$, $C|AB$ bipartitions but entangled with positive partial transpose in $A|BC$ bipartition. These imply that $E_{fin}>0$ while $E_{in}$ = $E_{com}$ = $0$. In this way, bound entanglement can be distributed by sending a separable qubit. Notice that by Proposition \ref{prop1}, the set $\{\ket{\psi_1}, \ket{\psi_3},\ket{\psi_5},\ket{\psi_7},\ket{\psi_9},\ket{s}\}$ must be extended to a complete orthogonal product basis in $B|AC$, $C|AB$ bipartitions. Therefore, the state $\rho_{ABC}$ must be separable in those two bipartitions. We now proceed to prove that the state $\rho_{ABC}$ is entangled in $A|BC$ bipartition. For this purpose, we first consider the state $\rho_{ABC}$ in $\mathbb{C}^3\otimes\mathbb{C}^4$. Because of the construction, the state $\rho_{ABC}$ must remain positive under partial transpose operation. To detect the entanglement of the state $\rho_{ABC}$ in that bipartition, we apply the technique described in Refs.~\cite{Sengupta13, Halder19-1}. We use the celebrated Choi map \cite{Choi75}, $\Lambda: M_3 \rightarrow M_3$, and a unitary operator $\mathcal{U}$ for the purpose of the detection of entanglement. Action of the Choi map (along with the unitary operator) is described below 
\begin{equation}\label{eq7}
\begin{array}{c}
(\Lambda_{\mathcal{U}}\otimes\mathbb{I})\rho_{A|BC} = (\Lambda\otimes\mathbb{I})(\mathcal{U}\otimes\mathbb{I})\rho_{A|BC}(\mathcal{U}\otimes\mathbb{I})^\dagger,\\[1.0 ex]

\Lambda: ((a_{ij})) \rightarrow 
\frac{1}{2}\left[
\begin{array}{ccc}
a_{11} + a_{22} &        -a_{12}       &        -a_{13}       \\
       -a_{21}       &  a_{22} + a_{33} &        -a_{23}       \\
    -a_{31}       &        -a_{32}       &  a_{33} + a_{11} \\
\end{array}\right],
\end{array}
\end{equation}
where $\mathcal{U}$ is a unitary operator. $\mathbb{I}$ is the identity operator acting on the subsystems $BC$ together. Here, we apply the following unitary operator 
\begin{equation}\label{eq8}
\mathcal{U} =\left[
\begin{array}{ccc}
 \frac{1}{2}  &      \frac{\sqrt{3}}{2}        &            0         \\ 
-\frac{\sqrt{3}}{2}  &      \frac{1}{2}        &            0         \\ 
    0        &          0             &            1         \\
\end{array}\right].
\end{equation}
One can easily check that the minimum eigenvalue of the operator $(\Lambda_{\mathcal{U}}\otimes\mathbb{I})\rho_{A|BC}$ is negative which confirms that the state $\rho_{ABC}$ is entangled in $A|BC$ bipartition. Thus, the protocol of distributing bound entanglement by sending a separable qubit is successfully accomplished. In the next section the conclusion is drawn.

\section{Conclusion}\label{sec6}
We have considered here the local distinguishability of multipartite orthogonal product states across the bipartitions and also across the multipartitions. A classification of tripartite systems has been given. Few multipartite cases have also been studied. These cases never occur for tripartite systems. This study is important to understand the efficient ways of resource sharing among the parties to distinguish a nonlocal set perfectly. We have also shown that if a set is locally indistinguishable in a particular bipartition then it can be used to distribute bound entanglement by sending a separable qubit. An explicit scenario has been shown considering minimum number of particles. However, there are open problems which can be considered for further studies. For example, a generalized scenario is required to construct, for the distribution of higher dimensional bound entanglement following the present method. Furthermore, it will also be interesting to apply the multipartite mixed states of the present kind in the quantum information processing protocols. 

\section*{Acknowledgment}
R. S. acknowledges the financial support from SERB MATRICS MTR/2017/000431.

\appendix
\begin{widetext}
\section{}\label{appA}
Before we present the detailed proof of Proposition \ref{prop3}, it is required to give the definition of a nontrivial orthogonality-preserving measurement (NOPM). In this regard also see Refs.~\cite{Groisman01, Walgate02}. \begin{definition}\label{def1} Suppose a set of orthogonal quantum states is given. To distinguish these states a measurement is performed. If after performing that measurement the post measurement states remain orthogonal then the measurement is an orthogonality-preserving measurement. Furthermore, if all the positive operator valued measure (POVM) elements describing that measurement are proportional to an identity operator then the measurement is a trivial orthogonality-preserving measurement. Otherwise, it is a nontrivial orthogonality-preserving measurement.\end{definition} The notations which we use in this appendix section, contain usual meaning as defined in Sec.~\ref{sec4}. We now proceed to prove that if $A_1$ and $A_2$ come together in a single location then they are not able to begin with a NOPM to distinguish the basis of Eq.~(\ref{eq5}). For this purpose we consider the basis in $A_1A_2|A_3|A_4$ tripartition ($\mathbb{C}^9\otimes\mathbb{C}^3\otimes\mathbb{C}^3$) and relabel it in the following way: $\ket{00}\rightarrow\ket{\bf0}$, $\ket{01}\rightarrow\ket{\bf1}$, $\ket{02}\rightarrow\ket{\bf2}$, $\ket{10}\rightarrow\ket{\bf3}$, $\ket{11}\rightarrow\ket{\bf4}$, $\ket{12}\rightarrow\ket{\bf5}$, $\ket{20}\rightarrow\ket{\bf6}$, $\ket{21}\rightarrow\ket{\bf7}$, $\ket{22}\rightarrow\ket{\bf8}$. The relabeled basis is given as 
\begin{equation}
\begin{array}{c}
\ket{\bf0}\ket{1}\ket{0\pm1},~
\ket{\bf0}\ket{2}\ket{0\pm2},~
\ket{\bf7}\ket{0}\ket{0\pm1},~
\ket{\bf4}\ket{2}\ket{0\pm1},~
\ket{\bf7}\ket{2}\ket{0\pm2},\\

\ket{\bf1}\ket{0\pm1}\ket{0},~
\ket{\bf2}\ket{0\pm2}\ket{0},~
\ket{\bf3}\ket{0\pm1}\ket{2},~
\ket{\bf5}\ket{0\pm1}\ket{1},~
\ket{\bf5}\ket{0\pm2}\ket{2},\\

\ket{{\bf3}\pm{\bf4}}\ket{0}\ket{0},~
\ket{{\bf6}\pm{\bf8}}\ket{0}\ket{0},~
\ket{{\bf0}\pm{\bf1}}\ket{2}\ket{1},~
\ket{{\bf6}\pm{\bf7}}\ket{1}\ket{1},~
\ket{{\bf6}\pm{\bf8}}\ket{2}\ket{1},\\

\ket{{\bf0}\pm{\bf3}}\ket{0}\ket{1},~
\ket{{\bf0}\pm{\bf6}}\ket{0}\ket{2},~
\ket{{\bf2}\pm{\bf5}}\ket{1}\ket{0},~
\ket{{\bf1}\pm{\bf4}}\ket{1}\ket{2},~
\ket{{\bf2}\pm{\bf8}}\ket{1}\ket{2},\\[1 ex]

\ket{\bf0}\ket{0}\ket{0},~

\ket{\bf0}\ket{1}\ket{2},~


\ket{\bf1}\ket{0}\ket{1},~
\ket{\bf1}\ket{0}\ket{2},~

\ket{\bf1}\ket{1}\ket{1},~

\ket{\bf1}\ket{2}\ket{0},\\
\ket{\bf1}\ket{2}\ket{2},~

\ket{\bf2}\ket{0}\ket{1},~
\ket{\bf2}\ket{0}\ket{2},~

\ket{\bf2}\ket{1}\ket{1},~

\ket{\bf2}\ket{2}\ket{1},~
\ket{\bf2}\ket{2}\ket{2},\\


\ket{\bf3}\ket{1}\ket{0},~
\ket{\bf3}\ket{1}\ket{1},~

\ket{\bf3}\ket{2}\ket{0},~
\ket{\bf3}\ket{2}\ket{1},~
\ket{\bf3}\ket{2}\ket{2},~

\ket{\bf4}\ket{0}\ket{1},\\
\ket{\bf4}\ket{0}\ket{2},~

\ket{\bf4}\ket{1}\ket{0},~
\ket{\bf4}\ket{1}\ket{1},~

\ket{\bf4}\ket{2}\ket{2},~

\ket{\bf5}\ket{0}\ket{0},~

\ket{\bf5}\ket{1}\ket{2},\\

\ket{\bf5}\ket{2}\ket{0},~
\ket{\bf5}\ket{2}\ket{1},~

\ket{\bf6}\ket{0}\ket{1},~

\ket{\bf6}\ket{1}\ket{0},~
\ket{\bf6}\ket{1}\ket{2},~

\ket{\bf6}\ket{2}\ket{0},\\
\ket{\bf6}\ket{2}\ket{2},~

\ket{\bf7}\ket{0}\ket{2},~

\ket{\bf7}\ket{1}\ket{0},~
\ket{\bf7}\ket{1}\ket{2},~

\ket{\bf7}\ket{2}\ket{1},~

\ket{\bf8}\ket{0}\ket{1},\\
\ket{\bf8}\ket{0}\ket{2},~

\ket{\bf8}\ket{1}\ket{0},~
\ket{\bf8}\ket{1}\ket{1},~

\ket{\bf8}\ket{2}\ket{0},~
\ket{\bf8}\ket{2}\ket{2}.
\end{array}
\end{equation}
An orthogonality-preserving measurement performed by $A_1$ and $A_2$ together can be described by a set of positive operator valued measure (POVM) elements $\{\Pi_i\}$, where $\sum_i\Pi_i =\mathbb{I}$. The matrix form of any such POVM element can be written in a basis $\{\ket{\bf0}, \ket{\bf1}, \ket{\bf2}, \ket{\bf3}, \ket{\bf4}, \ket{\bf5}, \ket{\bf6}, \ket{\bf7}, \ket{\bf8}\}$. Notice that the dimension of the joint system of $A_1$ and $A_2$ is nine and thus, it is obvious that the POVM elements which describe a measurement performed by $A_1$ and $A_2$ together, must be $9\times9$ matrices. The matrix form of any $\Pi_i$ = $\mbox{M}_i^\dagger\mbox{M}_i$ is given as the following
\begin{equation}
\begin{pmatrix}
a_{00} & a_{01} & a_{02} & a_{03} & a_{04} & a_{05} & a_{06} & a_{07} & a_{08} \\ 
a_{10} & a_{11} & a_{12} & a_{13} & a_{14} & a_{15} & a_{16} & a_{17} & a_{18} \\
a_{20} & a_{21} & a_{22} & a_{23} & a_{24} & a_{25} & a_{26} & a_{27} & a_{28} \\
a_{30} & a_{31} & a_{32} & a_{33} & a_{34} & a_{35} & a_{36} & a_{37} & a_{38} \\
a_{40} & a_{41} & a_{42} & a_{43} & a_{44} & a_{45} & a_{46} & a_{47} & a_{48} \\ 
a_{50} & a_{51} & a_{52} & a_{53} & a_{54} & a_{55} & a_{56} & a_{57} & a_{58} \\
a_{60} & a_{61} & a_{62} & a_{63} & a_{64} & a_{65} & a_{66} & a_{67} & a_{68} \\
a_{70} & a_{71} & a_{72} & a_{73} & a_{74} & a_{75} & a_{76} & a_{77} & a_{78} \\
a_{80} & a_{81} & a_{82} & a_{83} & a_{84} & a_{85} & a_{86} & a_{87} & a_{88} \\
\end{pmatrix}
\end{equation}
Keeping post measurement states orthogonal if $A_1$ and $A_2$ together want to perform a nontrivial measurement then not all the POVM elements $\Pi_i$ should be proportional to an identity operator. Now, we examine the off-diagonal entries of the above matrix. Because the post measurement states are orthogonal to each other so, $\langle{\bf0}|\langle1|\langle0+1|\Pi_i\otimes\mathbb{I}\otimes\mathbb{I}||{\bf1}\rangle|0+1\rangle|0\rangle$ = $\langle{\bf0}|\Pi_i|{\bf1}\rangle\langle1|0+1\rangle\langle0+1|0\rangle = 0$ $\Rightarrow$ $a_{01} = a_{10} = 0$. Repeating the same method, it is possible to show that all off-diagonal entries of the above matrix are zero. This is given in the Table \ref{tab1}.
\begin{table}[ht]
\caption{Off-diagonal entries}\label{tab1}
\centering
\begin{tabular}{|c|c|c||c|c|c|}
\hline
Sl. no. & states & entries & Sl. no. & states & entries \\
\hline\hline
(1) & $|{\bf0}\rangle|1\rangle|0+1\rangle, |{\bf1}\rangle|0+1\rangle|0\rangle$ & $a_{01}$ = $a_{10}$ = $0$ &
    
(2) & $|{\bf0}\rangle|1\rangle|0+1\rangle, |{\bf2}\rangle|1\rangle|1\rangle$ & $a_{02}$ = $a_{20}$ = $0$ \\\hline
    
(3) & $|{\bf0}\rangle|1\rangle|0+1\rangle, |{\bf3}\rangle|1\rangle|0\rangle$ & $a_{03}$ = $a_{30}$ = $0$ &
    
(4) & $|{\bf0}\rangle|1\rangle|0+1\rangle, |{\bf4}\rangle|1\rangle|0\rangle$ & $a_{04}$ = $a_{40}$ = $0$ \\\hline
    
(5) & $|{\bf0}\rangle|1\rangle|0+1\rangle, |{\bf5}\rangle|0+1\rangle|1\rangle$ & $a_{05}$ = $a_{50}$ = $0$ &
    
(6) & $|{\bf0}\rangle|1\rangle|0+1\rangle, |{\bf6}\rangle|1\rangle|0\rangle$ & $a_{06}$ = $a_{60}$ = $0$ \\\hline
    
(7) & $|{\bf0}\rangle|1\rangle|0+1\rangle, |{\bf7}\rangle|1\rangle|0\rangle$ & $a_{07}$ = $a_{70}$ = $0$ &
    
(8) & $|{\bf0}\rangle|1\rangle|0+1\rangle, |{\bf8}\rangle|1\rangle|0\rangle$ & $a_{08}$ = $a_{80}$ = $0$ \\\hline
    
(9) & $|{\bf1}\rangle|0+1\rangle|0\rangle, |{\bf2}\rangle|0+2\rangle|0\rangle$ & $a_{12}$ = $a_{21}$ = $0$ &
     
(10) & $|{\bf1}\rangle|0+1\rangle|0\rangle, |{\bf3}\rangle|1\rangle|0\rangle$ & $a_{13}$ = $a_{31}$ = $0$ \\\hline
     
(11) & $|{\bf1}\rangle|0+1\rangle|0\rangle, |{\bf4}\rangle|1\rangle|0\rangle$ & $a_{14}$ = $a_{41}$ = $0$ &
     
(12) & $|{\bf1}\rangle|0+1\rangle|0\rangle, |{\bf5}\rangle|0\rangle|0\rangle$ & $a_{15}$ = $a_{51}$ = $0$ \\\hline
     
(13) & $|{\bf1}\rangle|0+1\rangle|0\rangle, |{\bf6}\rangle|1\rangle|0\rangle$ & $a_{16}$ = $a_{61}$ = $0$ &
     
(14) & $|{\bf1}\rangle|0+1\rangle|0\rangle, |{\bf7}\rangle|1\rangle|0\rangle$ & $a_{17}$ = $a_{71}$ = $0$ \\\hline
     
(15) & $|{\bf1}\rangle|0+1\rangle|0\rangle, |{\bf8}\rangle|1\rangle|0\rangle$ & $a_{18}$ = $a_{81}$ = $0$ &
     
(16) & $|{\bf2}\rangle|0+2\rangle|0\rangle, |{\bf3}\rangle|2\rangle|0\rangle$ & $a_{23}$ = $a_{32}$ = $0$ \\\hline
     
(17) & $|{\bf2}\rangle|0+2\rangle|0\rangle, |{\bf4}\rangle|2\rangle|0+1\rangle$ & $a_{24}$ = $a_{42}$ = $0$ &
     
(18) & $|{\bf2}\rangle|0+2\rangle|0\rangle, |{\bf5}\rangle|2\rangle|0\rangle$ & $a_{25}$ = $a_{52}$ = $0$ \\\hline
     
(19) & $|{\bf2}\rangle|0+2\rangle|0\rangle, |{\bf6}\rangle|2\rangle|0\rangle$ & $a_{26}$ = $a_{62}$ = $0$ &
     
(20) & $|{\bf2}\rangle|0+2\rangle|0\rangle, |{\bf7}\rangle|2\rangle|0+2\rangle$ & $a_{27}$ = $a_{72}$ = $0$ \\\hline
     
(21) & $|{\bf2}\rangle|0+2\rangle|0\rangle, |{\bf8}\rangle|2\rangle|0\rangle$ & $a_{28}$ = $a_{82}$ = $0$ &
     
(22) & $|{\bf3}\rangle|0+1\rangle|2\rangle, |{\bf4}\rangle|0\rangle|2\rangle$ & $a_{34}$ = $a_{43}$ = $0$ \\\hline
     
(23) & $|{\bf3}\rangle|0+1\rangle|2\rangle, |{\bf5}\rangle|1\rangle|2\rangle$ & $a_{35}$ = $a_{53}$ = $0$ &
     
(24) & $|{\bf3}\rangle|0+1\rangle|2\rangle, |{\bf6}\rangle|1\rangle|2\rangle$ & $a_{36}$ = $a_{63}$ = $0$ \\\hline
     
(25) & $|{\bf3}\rangle|0+1\rangle|2\rangle, |{\bf7}\rangle|1\rangle|2\rangle$ & $a_{37}$ = $a_{73}$ = $0$ &
     
(26) & $|{\bf3}\rangle|0+1\rangle|2\rangle, |{\bf8}\rangle|0\rangle|2\rangle$ & $a_{38}$ = $a_{83}$ = $0$ \\\hline
     
(27) & $|{\bf4}\rangle|2\rangle|0+1\rangle, |{\bf5}\rangle|2\rangle|1\rangle$ & $a_{45}$ = $a_{54}$ = $0$ &
     
(28) & $|{\bf4}\rangle|2\rangle|0+1\rangle, |{\bf6}\rangle|2\rangle|0\rangle$ & $a_{46}$ = $a_{64}$ = $0$ \\\hline
     
(29) & $|{\bf4}\rangle|2\rangle|0+1\rangle, |{\bf7}\rangle|2\rangle|1\rangle$ & $a_{47}$ = $a_{74}$ = $0$ &
     
(30) & $|{\bf4}\rangle|2\rangle|0+1\rangle, |{\bf8}\rangle|2\rangle|0\rangle$ & $a_{48}$ = $a_{84}$ = $0$ \\\hline
     
(31) & $|{\bf5}\rangle|0+1\rangle|1\rangle, |{\bf6}\rangle|0\rangle|1\rangle$ & $a_{56}$ = $a_{65}$ = $0$ &
     
(32) & $|{\bf5}\rangle|0+1\rangle|1\rangle, |{\bf7}\rangle|0\rangle|0+1\rangle$ & $a_{57}$ = $a_{75}$ = $0$ \\\hline
     
(33) & $|{\bf5}\rangle|0+1\rangle|1\rangle, |{\bf8}\rangle|0\rangle|1\rangle$ & $a_{58}$ = $a_{85}$ = $0$ &
     
(34) & $|{\bf6}\rangle|1\rangle|0\rangle, |{\bf7}\rangle|1\rangle|0\rangle$ & $a_{67}$ = $a_{76}$ = $0$ \\\hline
     
(35) & $|{\bf6}\rangle|1\rangle|0\rangle, |{\bf8}\rangle|1\rangle|0\rangle$ & $a_{68}$ = $a_{86}$ = $0$ &
     
(36) & $|{\bf7}\rangle|1\rangle|0\rangle, |{\bf8}\rangle|1\rangle|0\rangle$ & $a_{78}$ = $a_{87}$ = $0$ \\\hline
\end{tabular}
\end{table}

Next, we examine the diagonal entries. Considering the inner product $\langle{\bf0}+{\bf1}|\langle2|\langle1|\Pi_i\otimes\mathbb{I}\otimes\mathbb{I}|{\bf0}-{\bf1}\rangle|2\rangle|1\rangle$ = $0$, we get $a_{00}$ = $a_{11}$. Repeating the same approach we get that all diagonal entries of the above matrix are equal. This is given in the Table \ref{tab2}. 
\begin{table}[ht]
\caption{Diagonal entries}\label{tab2}
\centering
\begin{tabular}{|c|c||c|c|}
\hline
states & entries & states & entries \\
\hline\hline
$\ket{{\bf3}\pm{\bf4}}\ket{0}\ket{0}$ & $a_{33}$ = $a_{44}$ &

$\ket{{\bf6}\pm{\bf8}}\ket{2}\ket{1}$ & $a_{66}$ = $a_{88}$ \\\hline

$\ket{{\bf0}\pm{\bf1}}\ket{2}\ket{1}$ & $a_{00}$ = $a_{11}$ &

$\ket{{\bf6}\pm{\bf7}}\ket{1}\ket{1}$ & $a_{66}$ = $a_{77}$ \\\hline

$\ket{{\bf0}\pm{\bf3}}\ket{0}\ket{1}$ & $a_{00}$ = $a_{33}$ &

$\ket{{\bf0}\pm{\bf6}}\ket{0}\ket{2}$ & $a_{00}$ = $a_{66}$ \\\hline

$\ket{{\bf2}\pm{\bf5}}\ket{1}\ket{0}$ & $a_{22}$ = $a_{55}$ &

$\ket{{\bf2}\pm{\bf8}}\ket{1}\ket{2}$ & $a_{22}$ = $a_{88}$ \\
\hline
\end{tabular}
\end{table}

Thus, it is proved that the POVM elements $\{\Pi_i\}$ are proportional to an identity operator. As a result of which $A_1$ and $A_2$ together cannot begin with a NOPM. This implies if $A_1$ and $A_2$ come together in a single location then it does not contribute to the local distinguishability of the considered basis (see the arguments given in Refs.~\cite{Groisman01, Walgate02}). Now, this holds true for all pairs of adjacent parties as the basis captures a particular type of symmetry. For non adjacent pairs of parties like $A_1A_3$ and $A_2A_4$, they cannot help for perfect discrimination by coming together and it is because of the first 16 states (see Eq.~(\ref{eq5}): first two columns in first row) of the basis. For the proof see Sec.~\ref{sec4}. Thus, if any two parties come together then also the basis shows local indistinguishability which implies that the basis is locally indistinguishable across all tripartitions. Here the proof is completed.
\end{widetext}

\bibliography{ref}

\begin{thebibliography}{91}%
\makeatletter
\providecommand \@ifxundefined [1]{%
 \@ifx{#1\undefined}
}%
\providecommand \@ifnum [1]{%
 \ifnum #1\expandafter \@firstoftwo
 \else \expandafter \@secondoftwo
 \fi
}%
\providecommand \@ifx [1]{%
 \ifx #1\expandafter \@firstoftwo
 \else \expandafter \@secondoftwo
 \fi
}%
\providecommand \natexlab [1]{#1}%
\providecommand \enquote  [1]{``#1''}%
\providecommand \bibnamefont  [1]{#1}%
\providecommand \bibfnamefont [1]{#1}%
\providecommand \citenamefont [1]{#1}%
\providecommand \href@noop [0]{\@secondoftwo}%
\providecommand \href [0]{\begingroup \@sanitize@url \@href}%
\providecommand \@href[1]{\@@startlink{#1}\@@href}%
\providecommand \@@href[1]{\endgroup#1\@@endlink}%
\providecommand \@sanitize@url [0]{\catcode `\\12\catcode `\$12\catcode
  `\&12\catcode `\#12\catcode `\^12\catcode `\_12\catcode `\%12\relax}%
\providecommand \@@startlink[1]{}%
\providecommand \@@endlink[0]{}%
\providecommand \url  [0]{\begingroup\@sanitize@url \@url }%
\providecommand \@url [1]{\endgroup\@href {#1}{\urlprefix }}%
\providecommand \urlprefix  [0]{URL }%
\providecommand \Eprint [0]{\href }%
\providecommand \doibase [0]{http://dx.doi.org/}%
\providecommand \selectlanguage [0]{\@gobble}%
\providecommand \bibinfo  [0]{\@secondoftwo}%
\providecommand \bibfield  [0]{\@secondoftwo}%
\providecommand \translation [1]{[#1]}%
\providecommand \BibitemOpen [0]{}%
\providecommand \bibitemStop [0]{}%
\providecommand \bibitemNoStop [0]{.\EOS\space}%
\providecommand \EOS [0]{\spacefactor3000\relax}%
\providecommand \BibitemShut  [1]{\csname bibitem#1\endcsname}%
\let\auto@bib@innerbib\@empty
\bibitem [{\citenamefont {Nielsen}\ and\ \citenamefont
  {Chuang}(2000)}]{Nielsen00}%
  \BibitemOpen
  \bibfield  {author} {\bibinfo {author} {\bibfnamefont {M.~A.}\ \bibnamefont
  {Nielsen}}\ and\ \bibinfo {author} {\bibfnamefont {I.~L.}\ \bibnamefont
  {Chuang}},\ }\href@noop {} {\bibfield  {journal} {\bibinfo  {journal} {{\it
  Quantum computation and quantum information}, Cambridge University Press}\ }
  (\bibinfo {year} {2000})}\BibitemShut {NoStop}%
\bibitem [{\citenamefont {Chitambar}\ \emph {et~al.}(2014)\citenamefont
  {Chitambar}, \citenamefont {Leung}, \citenamefont {Man\v{c}inska},
  \citenamefont {Ozols},\ and\ \citenamefont {Winter}}]{Chitambar14}%
  \BibitemOpen
  \bibfield  {author} {\bibinfo {author} {\bibfnamefont {E.}~\bibnamefont
  {Chitambar}}, \bibinfo {author} {\bibfnamefont {D.}~\bibnamefont {Leung}},
  \bibinfo {author} {\bibfnamefont {L.}~\bibnamefont {Man\v{c}inska}}, \bibinfo
  {author} {\bibfnamefont {M.}~\bibnamefont {Ozols}}, \ and\ \bibinfo {author}
  {\bibfnamefont {A.}~\bibnamefont {Winter}},\ }\href {\doibase
  https://doi.org/10.1007/s00220-014-1953-9} {\bibfield  {journal} {\bibinfo
  {journal} {Commun. Math. Phys.}\ }\textbf {\bibinfo {volume} {328}},\
  \bibinfo {pages} {303} (\bibinfo {year} {2014})}\BibitemShut {NoStop}%
\bibitem [{\citenamefont {Peres}\ and\ \citenamefont
  {Wootters}(1991)}]{Peres91}%
  \BibitemOpen
  \bibfield  {author} {\bibinfo {author} {\bibfnamefont {A.}~\bibnamefont
  {Peres}}\ and\ \bibinfo {author} {\bibfnamefont {W.~K.}\ \bibnamefont
  {Wootters}},\ }\href {\doibase 10.1103/PhysRevLett.66.1119} {\bibfield
  {journal} {\bibinfo  {journal} {Phys. Rev. Lett.}\ }\textbf {\bibinfo
  {volume} {66}},\ \bibinfo {pages} {1119} (\bibinfo {year}
  {1991})}\BibitemShut {NoStop}%
\bibitem [{\citenamefont {Walgate}\ \emph {et~al.}(2000)\citenamefont
  {Walgate}, \citenamefont {Short}, \citenamefont {Hardy},\ and\ \citenamefont
  {Vedral}}]{Walgate00}%
  \BibitemOpen
  \bibfield  {author} {\bibinfo {author} {\bibfnamefont {J.}~\bibnamefont
  {Walgate}}, \bibinfo {author} {\bibfnamefont {A.~J.}\ \bibnamefont {Short}},
  \bibinfo {author} {\bibfnamefont {L.}~\bibnamefont {Hardy}}, \ and\ \bibinfo
  {author} {\bibfnamefont {V.}~\bibnamefont {Vedral}},\ }\href {\doibase
  10.1103/PhysRevLett.85.4972} {\bibfield  {journal} {\bibinfo  {journal}
  {Phys. Rev. Lett.}\ }\textbf {\bibinfo {volume} {85}},\ \bibinfo {pages}
  {4972} (\bibinfo {year} {2000})}\BibitemShut {NoStop}%
\bibitem [{\citenamefont {Virmani}\ \emph {et~al.}(2001)\citenamefont
  {Virmani}, \citenamefont {Sacchi}, \citenamefont {Plenio},\ and\
  \citenamefont {Markham}}]{Virmani01}%
  \BibitemOpen
  \bibfield  {author} {\bibinfo {author} {\bibfnamefont {S.}~\bibnamefont
  {Virmani}}, \bibinfo {author} {\bibfnamefont {M.~F.}\ \bibnamefont {Sacchi}},
  \bibinfo {author} {\bibfnamefont {M.~B.}\ \bibnamefont {Plenio}}, \ and\
  \bibinfo {author} {\bibfnamefont {D.}~\bibnamefont {Markham}},\ }\href
  {\doibase https://doi.org/10.1016/S0375-9601(01)00484-4} {\bibfield
  {journal} {\bibinfo  {journal} {Phys. Lett. A.}\ }\textbf {\bibinfo {volume}
  {288}},\ \bibinfo {pages} {62} (\bibinfo {year} {2001})}\BibitemShut
  {NoStop}%
\bibitem [{\citenamefont {Ghosh}\ \emph {et~al.}(2001)\citenamefont {Ghosh},
  \citenamefont {Kar}, \citenamefont {Roy}, \citenamefont {Sen(De)},\ and\
  \citenamefont {Sen}}]{Ghosh01}%
  \BibitemOpen
  \bibfield  {author} {\bibinfo {author} {\bibfnamefont {S.}~\bibnamefont
  {Ghosh}}, \bibinfo {author} {\bibfnamefont {G.}~\bibnamefont {Kar}}, \bibinfo
  {author} {\bibfnamefont {A.}~\bibnamefont {Roy}}, \bibinfo {author}
  {\bibfnamefont {A.}~\bibnamefont {Sen(De)}}, \ and\ \bibinfo {author}
  {\bibfnamefont {U.}~\bibnamefont {Sen}},\ }\href {\doibase
  10.1103/PhysRevLett.87.277902} {\bibfield  {journal} {\bibinfo  {journal}
  {Phys. Rev. Lett.}\ }\textbf {\bibinfo {volume} {87}},\ \bibinfo {pages}
  {277902} (\bibinfo {year} {2001})}\BibitemShut {NoStop}%
\bibitem [{\citenamefont {Horodecki}\ \emph {et~al.}(2003)\citenamefont
  {Horodecki}, \citenamefont {Sen(De)}, \citenamefont {Sen},\ and\
  \citenamefont {Horodecki}}]{Horodecki03}%
  \BibitemOpen
  \bibfield  {author} {\bibinfo {author} {\bibfnamefont {M.}~\bibnamefont
  {Horodecki}}, \bibinfo {author} {\bibfnamefont {A.}~\bibnamefont {Sen(De)}},
  \bibinfo {author} {\bibfnamefont {U.}~\bibnamefont {Sen}}, \ and\ \bibinfo
  {author} {\bibfnamefont {K.}~\bibnamefont {Horodecki}},\ }\href {\doibase
  10.1103/PhysRevLett.90.047902} {\bibfield  {journal} {\bibinfo  {journal}
  {Phys. Rev. Lett.}\ }\textbf {\bibinfo {volume} {90}},\ \bibinfo {pages}
  {047902} (\bibinfo {year} {2003})}\BibitemShut {NoStop}%
\bibitem [{\citenamefont {Fan}(2004)}]{Fan04}%
  \BibitemOpen
  \bibfield  {author} {\bibinfo {author} {\bibfnamefont {H.}~\bibnamefont
  {Fan}},\ }\href {\doibase 10.1103/PhysRevLett.92.177905} {\bibfield
  {journal} {\bibinfo  {journal} {Phys. Rev. Lett.}\ }\textbf {\bibinfo
  {volume} {92}},\ \bibinfo {pages} {177905} (\bibinfo {year}
  {2004})}\BibitemShut {NoStop}%
\bibitem [{\citenamefont {Ghosh}\ \emph {et~al.}(2004)\citenamefont {Ghosh},
  \citenamefont {Kar}, \citenamefont {Roy},\ and\ \citenamefont
  {Sarkar}}]{Ghosh04}%
  \BibitemOpen
  \bibfield  {author} {\bibinfo {author} {\bibfnamefont {S.}~\bibnamefont
  {Ghosh}}, \bibinfo {author} {\bibfnamefont {G.}~\bibnamefont {Kar}}, \bibinfo
  {author} {\bibfnamefont {A.}~\bibnamefont {Roy}}, \ and\ \bibinfo {author}
  {\bibfnamefont {D.}~\bibnamefont {Sarkar}},\ }\href {\doibase
  10.1103/PhysRevA.70.022304} {\bibfield  {journal} {\bibinfo  {journal} {Phys.
  Rev. A}\ }\textbf {\bibinfo {volume} {70}},\ \bibinfo {pages} {022304}
  (\bibinfo {year} {2004})}\BibitemShut {NoStop}%
\bibitem [{\citenamefont {Nathanson}(2005)}]{Nathanson05}%
  \BibitemOpen
  \bibfield  {author} {\bibinfo {author} {\bibfnamefont {M.}~\bibnamefont
  {Nathanson}},\ }\href {\doibase https://doi.org/10.1063/1.1914731} {\bibfield
   {journal} {\bibinfo  {journal} {J. Math. Phys.}\ }\textbf {\bibinfo {volume}
  {46}},\ \bibinfo {pages} {062103} (\bibinfo {year} {2005})}\BibitemShut
  {NoStop}%
\bibitem [{\citenamefont {Watrous}(2005)}]{Watrous05}%
  \BibitemOpen
  \bibfield  {author} {\bibinfo {author} {\bibfnamefont {J.}~\bibnamefont
  {Watrous}},\ }\href {\doibase 10.1103/PhysRevLett.95.080505} {\bibfield
  {journal} {\bibinfo  {journal} {Phys. Rev. Lett.}\ }\textbf {\bibinfo
  {volume} {95}},\ \bibinfo {pages} {080505} (\bibinfo {year}
  {2005})}\BibitemShut {NoStop}%
\bibitem [{\citenamefont {Fan}(2007)}]{Fan07}%
  \BibitemOpen
  \bibfield  {author} {\bibinfo {author} {\bibfnamefont {H.}~\bibnamefont
  {Fan}},\ }\href {\doibase 10.1103/PhysRevA.75.014305} {\bibfield  {journal}
  {\bibinfo  {journal} {Phys. Rev. A}\ }\textbf {\bibinfo {volume} {75}},\
  \bibinfo {pages} {014305} (\bibinfo {year} {2007})}\BibitemShut {NoStop}%
\bibitem [{\citenamefont {Duan}\ \emph {et~al.}(2007)\citenamefont {Duan},
  \citenamefont {Feng}, \citenamefont {Ji},\ and\ \citenamefont
  {Ying}}]{Duan07}%
  \BibitemOpen
  \bibfield  {author} {\bibinfo {author} {\bibfnamefont {R.}~\bibnamefont
  {Duan}}, \bibinfo {author} {\bibfnamefont {Y.}~\bibnamefont {Feng}}, \bibinfo
  {author} {\bibfnamefont {Z.}~\bibnamefont {Ji}}, \ and\ \bibinfo {author}
  {\bibfnamefont {M.}~\bibnamefont {Ying}},\ }\href {\doibase
  10.1103/PhysRevLett.98.230502} {\bibfield  {journal} {\bibinfo  {journal}
  {Phys. Rev. Lett.}\ }\textbf {\bibinfo {volume} {98}},\ \bibinfo {pages}
  {230502} (\bibinfo {year} {2007})}\BibitemShut {NoStop}%
\bibitem [{\citenamefont {Bandyopadhyay}\ and\ \citenamefont
  {Walgate}(2009)}]{Bandyopadhyay09}%
  \BibitemOpen
  \bibfield  {author} {\bibinfo {author} {\bibfnamefont {S.}~\bibnamefont
  {Bandyopadhyay}}\ and\ \bibinfo {author} {\bibfnamefont {J.}~\bibnamefont
  {Walgate}},\ }\href {\doibase
  http://iopscience.iop.org/1751-8121/42/7/072002} {\bibfield  {journal}
  {\bibinfo  {journal} {J. Phys. A: Math. Theor.}\ }\textbf {\bibinfo {volume}
  {42}},\ \bibinfo {pages} {072002} (\bibinfo {year} {2009})}\BibitemShut
  {NoStop}%
\bibitem [{\citenamefont {Yu}\ \emph {et~al.}(2012)\citenamefont {Yu},
  \citenamefont {Duan},\ and\ \citenamefont {Ying}}]{Yu12}%
  \BibitemOpen
  \bibfield  {author} {\bibinfo {author} {\bibfnamefont {N.}~\bibnamefont
  {Yu}}, \bibinfo {author} {\bibfnamefont {R.}~\bibnamefont {Duan}}, \ and\
  \bibinfo {author} {\bibfnamefont {M.}~\bibnamefont {Ying}},\ }\href {\doibase
  10.1103/PhysRevLett.109.020506} {\bibfield  {journal} {\bibinfo  {journal}
  {Phys. Rev. Lett.}\ }\textbf {\bibinfo {volume} {109}},\ \bibinfo {pages}
  {020506} (\bibinfo {year} {2012})}\BibitemShut {NoStop}%
\bibitem [{\citenamefont {Markham}\ and\ \citenamefont
  {Sanders}(2008)}]{Markham08}%
  \BibitemOpen
  \bibfield  {author} {\bibinfo {author} {\bibfnamefont {D.}~\bibnamefont
  {Markham}}\ and\ \bibinfo {author} {\bibfnamefont {B.~C.}\ \bibnamefont
  {Sanders}},\ }\href {\doibase 10.1103/PhysRevA.78.042309} {\bibfield
  {journal} {\bibinfo  {journal} {Phys. Rev. A}\ }\textbf {\bibinfo {volume}
  {78}},\ \bibinfo {pages} {042309} (\bibinfo {year} {2008})}\BibitemShut
  {NoStop}%
\bibitem [{\citenamefont {Terhal}\ \emph {et~al.}(2001)\citenamefont {Terhal},
  \citenamefont {DiVincenzo},\ and\ \citenamefont {Leung}}]{Terhal01}%
  \BibitemOpen
  \bibfield  {author} {\bibinfo {author} {\bibfnamefont {B.~M.}\ \bibnamefont
  {Terhal}}, \bibinfo {author} {\bibfnamefont {D.~P.}\ \bibnamefont
  {DiVincenzo}}, \ and\ \bibinfo {author} {\bibfnamefont {D.~W.}\ \bibnamefont
  {Leung}},\ }\href {\doibase 10.1103/PhysRevLett.86.5807} {\bibfield
  {journal} {\bibinfo  {journal} {Phys. Rev. Lett.}\ }\textbf {\bibinfo
  {volume} {86}},\ \bibinfo {pages} {5807} (\bibinfo {year}
  {2001})}\BibitemShut {NoStop}%
\bibitem [{\citenamefont {Eggeling}\ and\ \citenamefont
  {Werner}(2002)}]{Eggeling02}%
  \BibitemOpen
  \bibfield  {author} {\bibinfo {author} {\bibfnamefont {T.}~\bibnamefont
  {Eggeling}}\ and\ \bibinfo {author} {\bibfnamefont {R.~F.}\ \bibnamefont
  {Werner}},\ }\href {\doibase 10.1103/PhysRevLett.89.097905} {\bibfield
  {journal} {\bibinfo  {journal} {Phys. Rev. Lett.}\ }\textbf {\bibinfo
  {volume} {89}},\ \bibinfo {pages} {097905} (\bibinfo {year}
  {2002})}\BibitemShut {NoStop}%
\bibitem [{\citenamefont {Bennett}\ \emph
  {et~al.}(1999{\natexlab{a}})\citenamefont {Bennett}, \citenamefont
  {DiVincenzo}, \citenamefont {Fuchs}, \citenamefont {Mor}, \citenamefont
  {Rains}, \citenamefont {Shor}, \citenamefont {Smolin},\ and\ \citenamefont
  {Wootters}}]{Bennett99-1}%
  \BibitemOpen
  \bibfield  {author} {\bibinfo {author} {\bibfnamefont {C.~H.}\ \bibnamefont
  {Bennett}}, \bibinfo {author} {\bibfnamefont {D.~P.}\ \bibnamefont
  {DiVincenzo}}, \bibinfo {author} {\bibfnamefont {C.~A.}\ \bibnamefont
  {Fuchs}}, \bibinfo {author} {\bibfnamefont {T.}~\bibnamefont {Mor}}, \bibinfo
  {author} {\bibfnamefont {E.}~\bibnamefont {Rains}}, \bibinfo {author}
  {\bibfnamefont {P.~W.}\ \bibnamefont {Shor}}, \bibinfo {author}
  {\bibfnamefont {J.~A.}\ \bibnamefont {Smolin}}, \ and\ \bibinfo {author}
  {\bibfnamefont {W.~K.}\ \bibnamefont {Wootters}},\ }\href {\doibase
  10.1103/PhysRevA.59.1070} {\bibfield  {journal} {\bibinfo  {journal} {Phys.
  Rev. A}\ }\textbf {\bibinfo {volume} {59}},\ \bibinfo {pages} {1070}
  (\bibinfo {year} {1999}{\natexlab{a}})}\BibitemShut {NoStop}%
\bibitem [{\citenamefont {Bennett}\ \emph
  {et~al.}(1999{\natexlab{b}})\citenamefont {Bennett}, \citenamefont
  {DiVincenzo}, \citenamefont {Mor}, \citenamefont {Shor}, \citenamefont
  {Smolin},\ and\ \citenamefont {Terhal}}]{Bennett99}%
  \BibitemOpen
  \bibfield  {author} {\bibinfo {author} {\bibfnamefont {C.~H.}\ \bibnamefont
  {Bennett}}, \bibinfo {author} {\bibfnamefont {D.~P.}\ \bibnamefont
  {DiVincenzo}}, \bibinfo {author} {\bibfnamefont {T.}~\bibnamefont {Mor}},
  \bibinfo {author} {\bibfnamefont {P.~W.}\ \bibnamefont {Shor}}, \bibinfo
  {author} {\bibfnamefont {J.~A.}\ \bibnamefont {Smolin}}, \ and\ \bibinfo
  {author} {\bibfnamefont {B.~M.}\ \bibnamefont {Terhal}},\ }\href {\doibase
  10.1103/PhysRevLett.82.5385} {\bibfield  {journal} {\bibinfo  {journal}
  {Phys. Rev. Lett.}\ }\textbf {\bibinfo {volume} {82}},\ \bibinfo {pages}
  {5385} (\bibinfo {year} {1999}{\natexlab{b}})}\BibitemShut {NoStop}%
\bibitem [{\citenamefont {DiVincenzo}\ \emph {et~al.}(2003)\citenamefont
  {DiVincenzo}, \citenamefont {Mor}, \citenamefont {Shor}, \citenamefont
  {Smolin},\ and\ \citenamefont {Terhal}}]{DiVincenzo03}%
  \BibitemOpen
  \bibfield  {author} {\bibinfo {author} {\bibfnamefont {D.~P.}\ \bibnamefont
  {DiVincenzo}}, \bibinfo {author} {\bibfnamefont {T.}~\bibnamefont {Mor}},
  \bibinfo {author} {\bibfnamefont {P.~W.}\ \bibnamefont {Shor}}, \bibinfo
  {author} {\bibfnamefont {J.~A.}\ \bibnamefont {Smolin}}, \ and\ \bibinfo
  {author} {\bibfnamefont {B.~M.}\ \bibnamefont {Terhal}},\ }\href {\doibase
  10.1007/s00220-003-0877-6} {\bibfield  {journal} {\bibinfo  {journal}
  {Commun. Math. Phys.}\ }\textbf {\bibinfo {volume} {238}},\ \bibinfo {pages}
  {379} (\bibinfo {year} {2003})}\BibitemShut {NoStop}%
\bibitem [{\citenamefont {Groisman}\ and\ \citenamefont
  {Vaidman}(2001)}]{Groisman01}%
  \BibitemOpen
  \bibfield  {author} {\bibinfo {author} {\bibfnamefont {B.}~\bibnamefont
  {Groisman}}\ and\ \bibinfo {author} {\bibfnamefont {L.}~\bibnamefont
  {Vaidman}},\ }\href {\doibase http://dx.doi.org/10.1088/0305-4470/34/35/313}
  {\bibfield  {journal} {\bibinfo  {journal} {J. Phys. A: Math. Gen.}\ }\textbf
  {\bibinfo {volume} {34}},\ \bibinfo {pages} {6881} (\bibinfo {year}
  {2001})}\BibitemShut {NoStop}%
\bibitem [{\citenamefont {Walgate}\ and\ \citenamefont
  {Hardy}(2002)}]{Walgate02}%
  \BibitemOpen
  \bibfield  {author} {\bibinfo {author} {\bibfnamefont {J.}~\bibnamefont
  {Walgate}}\ and\ \bibinfo {author} {\bibfnamefont {L.}~\bibnamefont
  {Hardy}},\ }\href {\doibase 10.1103/PhysRevLett.89.147901} {\bibfield
  {journal} {\bibinfo  {journal} {Phys. Rev. Lett.}\ }\textbf {\bibinfo
  {volume} {89}},\ \bibinfo {pages} {147901} (\bibinfo {year}
  {2002})}\BibitemShut {NoStop}%
\bibitem [{\citenamefont {De~Rinaldis}(2004)}]{Rinaldis04}%
  \BibitemOpen
  \bibfield  {author} {\bibinfo {author} {\bibfnamefont {S.}~\bibnamefont
  {De~Rinaldis}},\ }\href {\doibase 10.1103/PhysRevA.70.022309} {\bibfield
  {journal} {\bibinfo  {journal} {Phys. Rev. A}\ }\textbf {\bibinfo {volume}
  {70}},\ \bibinfo {pages} {022309} (\bibinfo {year} {2004})}\BibitemShut
  {NoStop}%
\bibitem [{\citenamefont {Niset}\ and\ \citenamefont {Cerf}(2006)}]{Niset06}%
  \BibitemOpen
  \bibfield  {author} {\bibinfo {author} {\bibfnamefont {J.}~\bibnamefont
  {Niset}}\ and\ \bibinfo {author} {\bibfnamefont {N.~J.}\ \bibnamefont
  {Cerf}},\ }\href {\doibase 10.1103/PhysRevA.74.052103} {\bibfield  {journal}
  {\bibinfo  {journal} {Phys. Rev. A}\ }\textbf {\bibinfo {volume} {74}},\
  \bibinfo {pages} {052103} (\bibinfo {year} {2006})}\BibitemShut {NoStop}%
\bibitem [{\citenamefont {Ye}\ \emph {et~al.}(2007)\citenamefont {Ye},
  \citenamefont {Jiang}, \citenamefont {Chen}, \citenamefont {Zhang},
  \citenamefont {Zhou},\ and\ \citenamefont {Guo}}]{Ye07}%
  \BibitemOpen
  \bibfield  {author} {\bibinfo {author} {\bibfnamefont {M.-Y.}\ \bibnamefont
  {Ye}}, \bibinfo {author} {\bibfnamefont {W.}~\bibnamefont {Jiang}}, \bibinfo
  {author} {\bibfnamefont {P.-X.}\ \bibnamefont {Chen}}, \bibinfo {author}
  {\bibfnamefont {Y.-S.}\ \bibnamefont {Zhang}}, \bibinfo {author}
  {\bibfnamefont {Z.-W.}\ \bibnamefont {Zhou}}, \ and\ \bibinfo {author}
  {\bibfnamefont {G.-C.}\ \bibnamefont {Guo}},\ }\href {\doibase
  10.1103/PhysRevA.76.032329} {\bibfield  {journal} {\bibinfo  {journal} {Phys.
  Rev. A}\ }\textbf {\bibinfo {volume} {76}},\ \bibinfo {pages} {032329}
  (\bibinfo {year} {2007})}\BibitemShut {NoStop}%
\bibitem [{\citenamefont {{Feng}}\ and\ \citenamefont {{Shi}}(2009)}]{Feng09}%
  \BibitemOpen
  \bibfield  {author} {\bibinfo {author} {\bibfnamefont {Y.}~\bibnamefont
  {{Feng}}}\ and\ \bibinfo {author} {\bibfnamefont {Y.}~\bibnamefont {{Shi}}},\
  }\href {\doibase 10.1109/TIT.2009.2018330} {\bibfield  {journal} {\bibinfo
  {journal} {IEEE Trans. Inf. Theory}\ }\textbf {\bibinfo {volume} {55}},\
  \bibinfo {pages} {2799} (\bibinfo {year} {2009})}\BibitemShut {NoStop}%
\bibitem [{\citenamefont {Duan}\ \emph {et~al.}(2010)\citenamefont {Duan},
  \citenamefont {Xin},\ and\ \citenamefont {Ying}}]{Duan10}%
  \BibitemOpen
  \bibfield  {author} {\bibinfo {author} {\bibfnamefont {R.}~\bibnamefont
  {Duan}}, \bibinfo {author} {\bibfnamefont {Y.}~\bibnamefont {Xin}}, \ and\
  \bibinfo {author} {\bibfnamefont {M.}~\bibnamefont {Ying}},\ }\href {\doibase
  10.1103/PhysRevA.81.032329} {\bibfield  {journal} {\bibinfo  {journal} {Phys.
  Rev. A}\ }\textbf {\bibinfo {volume} {81}},\ \bibinfo {pages} {032329}
  (\bibinfo {year} {2010})}\BibitemShut {NoStop}%
\bibitem [{\citenamefont {Yang}\ \emph {et~al.}(2013)\citenamefont {Yang},
  \citenamefont {Gao}, \citenamefont {Tian}, \citenamefont {Cao},\ and\
  \citenamefont {Wen}}]{Yang13}%
  \BibitemOpen
  \bibfield  {author} {\bibinfo {author} {\bibfnamefont {Y.-H.}\ \bibnamefont
  {Yang}}, \bibinfo {author} {\bibfnamefont {F.}~\bibnamefont {Gao}}, \bibinfo
  {author} {\bibfnamefont {G.-J.}\ \bibnamefont {Tian}}, \bibinfo {author}
  {\bibfnamefont {T.-Q.}\ \bibnamefont {Cao}}, \ and\ \bibinfo {author}
  {\bibfnamefont {Q.-Y.}\ \bibnamefont {Wen}},\ }\href {\doibase
  10.1103/PhysRevA.88.024301} {\bibfield  {journal} {\bibinfo  {journal} {Phys.
  Rev. A}\ }\textbf {\bibinfo {volume} {88}},\ \bibinfo {pages} {024301}
  (\bibinfo {year} {2013})}\BibitemShut {NoStop}%
\bibitem [{\citenamefont {Childs}\ \emph {et~al.}(2013)\citenamefont {Childs},
  \citenamefont {Leung}, \citenamefont {Man{\v{c}}inska},\ and\ \citenamefont
  {Ozols}}]{Childs13}%
  \BibitemOpen
  \bibfield  {author} {\bibinfo {author} {\bibfnamefont {A.~M.}\ \bibnamefont
  {Childs}}, \bibinfo {author} {\bibfnamefont {D.}~\bibnamefont {Leung}},
  \bibinfo {author} {\bibfnamefont {L.}~\bibnamefont {Man{\v{c}}inska}}, \ and\
  \bibinfo {author} {\bibfnamefont {M.}~\bibnamefont {Ozols}},\ }\href
  {\doibase 10.1007/s00220-013-1784-0} {\bibfield  {journal} {\bibinfo
  {journal} {Commun. Math. Phys.}\ }\textbf {\bibinfo {volume} {323}},\
  \bibinfo {pages} {1121} (\bibinfo {year} {2013})}\BibitemShut {NoStop}%
\bibitem [{\citenamefont {Zhang}\ \emph {et~al.}(2014)\citenamefont {Zhang},
  \citenamefont {Gao}, \citenamefont {Tian}, \citenamefont {Cao},\ and\
  \citenamefont {Wen}}]{Zhang14}%
  \BibitemOpen
  \bibfield  {author} {\bibinfo {author} {\bibfnamefont {Z.-C.}\ \bibnamefont
  {Zhang}}, \bibinfo {author} {\bibfnamefont {F.}~\bibnamefont {Gao}}, \bibinfo
  {author} {\bibfnamefont {G.-J.}\ \bibnamefont {Tian}}, \bibinfo {author}
  {\bibfnamefont {T.-Q.}\ \bibnamefont {Cao}}, \ and\ \bibinfo {author}
  {\bibfnamefont {Q.-Y.}\ \bibnamefont {Wen}},\ }\href {\doibase
  10.1103/PhysRevA.90.022313} {\bibfield  {journal} {\bibinfo  {journal} {Phys.
  Rev. A}\ }\textbf {\bibinfo {volume} {90}},\ \bibinfo {pages} {022313}
  (\bibinfo {year} {2014})}\BibitemShut {NoStop}%
\bibitem [{\citenamefont {Yu}\ and\ \citenamefont {Oh}(2015)}]{Yu15}%
  \BibitemOpen
  \bibfield  {author} {\bibinfo {author} {\bibfnamefont {S.~X.}\ \bibnamefont
  {Yu}}\ and\ \bibinfo {author} {\bibfnamefont {C.~H.}\ \bibnamefont {Oh}},\
  }\href {https://arxiv.org/abs/1502.01274} {\bibfield  {journal} {\bibinfo
  {journal} {arXiv:1502.01274 [quant-ph]}\ } (\bibinfo {year}
  {2015})}\BibitemShut {NoStop}%
\bibitem [{\citenamefont {Zhang}\ \emph {et~al.}(2015)\citenamefont {Zhang},
  \citenamefont {Gao}, \citenamefont {Qin}, \citenamefont {Yang},\ and\
  \citenamefont {Wen}}]{Zhang15}%
  \BibitemOpen
  \bibfield  {author} {\bibinfo {author} {\bibfnamefont {Z.-C.}\ \bibnamefont
  {Zhang}}, \bibinfo {author} {\bibfnamefont {F.}~\bibnamefont {Gao}}, \bibinfo
  {author} {\bibfnamefont {S.-J.}\ \bibnamefont {Qin}}, \bibinfo {author}
  {\bibfnamefont {Y.-H.}\ \bibnamefont {Yang}}, \ and\ \bibinfo {author}
  {\bibfnamefont {Q.-Y.}\ \bibnamefont {Wen}},\ }\href {\doibase
  10.1103/PhysRevA.92.012332} {\bibfield  {journal} {\bibinfo  {journal} {Phys.
  Rev. A}\ }\textbf {\bibinfo {volume} {92}},\ \bibinfo {pages} {012332}
  (\bibinfo {year} {2015})}\BibitemShut {NoStop}%
\bibitem [{\citenamefont {Wang}\ \emph {et~al.}(2015)\citenamefont {Wang},
  \citenamefont {Li}, \citenamefont {Zheng},\ and\ \citenamefont
  {Fei}}]{Wang15}%
  \BibitemOpen
  \bibfield  {author} {\bibinfo {author} {\bibfnamefont {Y.-L.}\ \bibnamefont
  {Wang}}, \bibinfo {author} {\bibfnamefont {M.-S.}\ \bibnamefont {Li}},
  \bibinfo {author} {\bibfnamefont {Z.-J.}\ \bibnamefont {Zheng}}, \ and\
  \bibinfo {author} {\bibfnamefont {S.-M.}\ \bibnamefont {Fei}},\ }\href
  {\doibase 10.1103/PhysRevA.92.032313} {\bibfield  {journal} {\bibinfo
  {journal} {Phys. Rev. A}\ }\textbf {\bibinfo {volume} {92}},\ \bibinfo
  {pages} {032313} (\bibinfo {year} {2015})}\BibitemShut {NoStop}%
\bibitem [{\citenamefont {Chen}\ and\ \citenamefont {Johnston}(2015)}]{Chen15}%
  \BibitemOpen
  \bibfield  {author} {\bibinfo {author} {\bibfnamefont {J.}~\bibnamefont
  {Chen}}\ and\ \bibinfo {author} {\bibfnamefont {N.}~\bibnamefont
  {Johnston}},\ }\href {\doibase 10.1007/s00220-014-2186-7} {\bibfield
  {journal} {\bibinfo  {journal} {Commun. Math. Phys.}\ }\textbf {\bibinfo
  {volume} {333}},\ \bibinfo {pages} {351} (\bibinfo {year}
  {2015})}\BibitemShut {NoStop}%
\bibitem [{\citenamefont {Yang}\ \emph {et~al.}(2015)\citenamefont {Yang},
  \citenamefont {Gao}, \citenamefont {Xu}, \citenamefont {Zuo}, \citenamefont
  {Zhang},\ and\ \citenamefont {Wen}}]{Yang15}%
  \BibitemOpen
  \bibfield  {author} {\bibinfo {author} {\bibfnamefont {Y.-H.}\ \bibnamefont
  {Yang}}, \bibinfo {author} {\bibfnamefont {F.}~\bibnamefont {Gao}}, \bibinfo
  {author} {\bibfnamefont {G.-B.}\ \bibnamefont {Xu}}, \bibinfo {author}
  {\bibfnamefont {H.-J.}\ \bibnamefont {Zuo}}, \bibinfo {author} {\bibfnamefont
  {Z.-C.}\ \bibnamefont {Zhang}}, \ and\ \bibinfo {author} {\bibfnamefont
  {Q.-Y.}\ \bibnamefont {Wen}},\ }\href {\doibase
  https://dx.doi.org/10.1038/srep11963} {\bibfield  {journal} {\bibinfo
  {journal} {Sci. Rep.}\ }\textbf {\bibinfo {volume} {5}},\ \bibinfo {pages}
  {11963} (\bibinfo {year} {2015})}\BibitemShut {NoStop}%
\bibitem [{\citenamefont {Zhang}\ \emph
  {et~al.}(2016{\natexlab{a}})\citenamefont {Zhang}, \citenamefont {Gao},
  \citenamefont {Cao}, \citenamefont {Qin},\ and\ \citenamefont
  {Wen}}]{Zhang16}%
  \BibitemOpen
  \bibfield  {author} {\bibinfo {author} {\bibfnamefont {Z.-C.}\ \bibnamefont
  {Zhang}}, \bibinfo {author} {\bibfnamefont {F.}~\bibnamefont {Gao}}, \bibinfo
  {author} {\bibfnamefont {Y.}~\bibnamefont {Cao}}, \bibinfo {author}
  {\bibfnamefont {S.-J.}\ \bibnamefont {Qin}}, \ and\ \bibinfo {author}
  {\bibfnamefont {Q.-Y.}\ \bibnamefont {Wen}},\ }\href {\doibase
  10.1103/PhysRevA.93.012314} {\bibfield  {journal} {\bibinfo  {journal} {Phys.
  Rev. A}\ }\textbf {\bibinfo {volume} {93}},\ \bibinfo {pages} {012314}
  (\bibinfo {year} {2016}{\natexlab{a}})}\BibitemShut {NoStop}%
\bibitem [{\citenamefont {Xu}\ \emph {et~al.}(2016{\natexlab{a}})\citenamefont
  {Xu}, \citenamefont {Wen}, \citenamefont {Qin}, \citenamefont {Yang},\ and\
  \citenamefont {Gao}}]{Xu16}%
  \BibitemOpen
  \bibfield  {author} {\bibinfo {author} {\bibfnamefont {G.-B.}\ \bibnamefont
  {Xu}}, \bibinfo {author} {\bibfnamefont {Q.-Y.}\ \bibnamefont {Wen}},
  \bibinfo {author} {\bibfnamefont {S.-J.}\ \bibnamefont {Qin}}, \bibinfo
  {author} {\bibfnamefont {Y.-H.}\ \bibnamefont {Yang}}, \ and\ \bibinfo
  {author} {\bibfnamefont {F.}~\bibnamefont {Gao}},\ }\href {\doibase
  10.1103/PhysRevA.93.032341} {\bibfield  {journal} {\bibinfo  {journal} {Phys.
  Rev. A}\ }\textbf {\bibinfo {volume} {93}},\ \bibinfo {pages} {032341}
  (\bibinfo {year} {2016}{\natexlab{a}})}\BibitemShut {NoStop}%
\bibitem [{\citenamefont {Zhang}\ \emph
  {et~al.}(2016{\natexlab{b}})\citenamefont {Zhang}, \citenamefont {Tan},
  \citenamefont {Weng},\ and\ \citenamefont {Li}}]{Zhang16-1}%
  \BibitemOpen
  \bibfield  {author} {\bibinfo {author} {\bibfnamefont {X.}~\bibnamefont
  {Zhang}}, \bibinfo {author} {\bibfnamefont {X.}~\bibnamefont {Tan}}, \bibinfo
  {author} {\bibfnamefont {J.}~\bibnamefont {Weng}}, \ and\ \bibinfo {author}
  {\bibfnamefont {Y.}~\bibnamefont {Li}},\ }\href {\doibase
  https://dx.doi.org/10.1038/srep28864} {\bibfield  {journal} {\bibinfo
  {journal} {Sci. Rep.}\ }\textbf {\bibinfo {volume} {6}},\ \bibinfo {pages}
  {28864} (\bibinfo {year} {2016}{\natexlab{b}})}\BibitemShut {NoStop}%
\bibitem [{\citenamefont {Xu}\ \emph {et~al.}(2016{\natexlab{b}})\citenamefont
  {Xu}, \citenamefont {Yang}, \citenamefont {Wen}, \citenamefont {Qin},\ and\
  \citenamefont {Gao}}]{Xu16-1}%
  \BibitemOpen
  \bibfield  {author} {\bibinfo {author} {\bibfnamefont {G.-B.}\ \bibnamefont
  {Xu}}, \bibinfo {author} {\bibfnamefont {Y.-H.}\ \bibnamefont {Yang}},
  \bibinfo {author} {\bibfnamefont {Q.-Y.}\ \bibnamefont {Wen}}, \bibinfo
  {author} {\bibfnamefont {S.-J.}\ \bibnamefont {Qin}}, \ and\ \bibinfo
  {author} {\bibfnamefont {F.}~\bibnamefont {Gao}},\ }\href {\doibase
  https://dx.doi.org/10.1038/srep31048} {\bibfield  {journal} {\bibinfo
  {journal} {Sci. Rep.}\ }\textbf {\bibinfo {volume} {6}},\ \bibinfo {pages}
  {31048} (\bibinfo {year} {2016}{\natexlab{b}})}\BibitemShut {NoStop}%
\bibitem [{\citenamefont {Wang}\ \emph {et~al.}(2017)\citenamefont {Wang},
  \citenamefont {Li}, \citenamefont {Fei},\ and\ \citenamefont
  {Zheng}}]{Wang17}%
  \BibitemOpen
  \bibfield  {author} {\bibinfo {author} {\bibfnamefont {Y.-L.}\ \bibnamefont
  {Wang}}, \bibinfo {author} {\bibfnamefont {M.-S.}\ \bibnamefont {Li}},
  \bibinfo {author} {\bibfnamefont {S.-M.}\ \bibnamefont {Fei}}, \ and\
  \bibinfo {author} {\bibfnamefont {Z.-J.}\ \bibnamefont {Zheng}},\ }\href
  {https://arxiv.org/abs/1703.06542} {\bibfield  {journal} {\bibinfo  {journal}
  {arXiv:1703.06542 [quant-ph]}\ } (\bibinfo {year} {2017})}\BibitemShut
  {NoStop}%
\bibitem [{\citenamefont {Zhang}\ \emph
  {et~al.}(2017{\natexlab{a}})\citenamefont {Zhang}, \citenamefont {Zhang},
  \citenamefont {Gao}, \citenamefont {Wen},\ and\ \citenamefont
  {Oh}}]{Zhang17}%
  \BibitemOpen
  \bibfield  {author} {\bibinfo {author} {\bibfnamefont {Z.-C.}\ \bibnamefont
  {Zhang}}, \bibinfo {author} {\bibfnamefont {K.-J.}\ \bibnamefont {Zhang}},
  \bibinfo {author} {\bibfnamefont {F.}~\bibnamefont {Gao}}, \bibinfo {author}
  {\bibfnamefont {Q.-Y.}\ \bibnamefont {Wen}}, \ and\ \bibinfo {author}
  {\bibfnamefont {C.~H.}\ \bibnamefont {Oh}},\ }\href {\doibase
  10.1103/PhysRevA.95.052344} {\bibfield  {journal} {\bibinfo  {journal} {Phys.
  Rev. A}\ }\textbf {\bibinfo {volume} {95}},\ \bibinfo {pages} {052344}
  (\bibinfo {year} {2017}{\natexlab{a}})}\BibitemShut {NoStop}%
\bibitem [{\citenamefont {Xu}\ \emph {et~al.}(2017)\citenamefont {Xu},
  \citenamefont {Wen}, \citenamefont {Gao}, \citenamefont {Qin},\ and\
  \citenamefont {Zuo}}]{Xu17}%
  \BibitemOpen
  \bibfield  {author} {\bibinfo {author} {\bibfnamefont {G.-B.}\ \bibnamefont
  {Xu}}, \bibinfo {author} {\bibfnamefont {Q.-Y.}\ \bibnamefont {Wen}},
  \bibinfo {author} {\bibfnamefont {F.}~\bibnamefont {Gao}}, \bibinfo {author}
  {\bibfnamefont {S.-J.}\ \bibnamefont {Qin}}, \ and\ \bibinfo {author}
  {\bibfnamefont {H.-J.}\ \bibnamefont {Zuo}},\ }\href {\doibase
  10.1007/s11128-017-1725-5} {\bibfield  {journal} {\bibinfo  {journal} {Quant.
  Info. Proc.}\ }\textbf {\bibinfo {volume} {16}},\ \bibinfo {pages} {276}
  (\bibinfo {year} {2017})}\BibitemShut {NoStop}%
\bibitem [{\citenamefont {Wang}\ \emph {et~al.}(2016)\citenamefont {Wang},
  \citenamefont {Li}, \citenamefont {Zheng},\ and\ \citenamefont
  {Fei}}]{Wang16}%
  \BibitemOpen
  \bibfield  {author} {\bibinfo {author} {\bibfnamefont {Y.-L.}\ \bibnamefont
  {Wang}}, \bibinfo {author} {\bibfnamefont {M.-S.}\ \bibnamefont {Li}},
  \bibinfo {author} {\bibfnamefont {Z.-J.}\ \bibnamefont {Zheng}}, \ and\
  \bibinfo {author} {\bibfnamefont {S.-M.}\ \bibnamefont {Fei}},\ }\href
  {\doibase 10.1007/s11128-016-1477-7} {\bibfield  {journal} {\bibinfo
  {journal} {Quant. Info. Proc.}\ }\textbf {\bibinfo {volume} {16}},\ \bibinfo
  {pages} {5} (\bibinfo {year} {2016})}\BibitemShut {NoStop}%
\bibitem [{\citenamefont {Zhang}\ \emph
  {et~al.}(2017{\natexlab{b}})\citenamefont {Zhang}, \citenamefont {Weng},
  \citenamefont {Tan},\ and\ \citenamefont {Luo}}]{Zhang17-1}%
  \BibitemOpen
  \bibfield  {author} {\bibinfo {author} {\bibfnamefont {X.}~\bibnamefont
  {Zhang}}, \bibinfo {author} {\bibfnamefont {J.}~\bibnamefont {Weng}},
  \bibinfo {author} {\bibfnamefont {X.}~\bibnamefont {Tan}}, \ and\ \bibinfo
  {author} {\bibfnamefont {W.}~\bibnamefont {Luo}},\ }\href {\doibase
  10.1007/s11128-017-1616-9} {\bibfield  {journal} {\bibinfo  {journal} {Quant.
  Info. Proc.}\ }\textbf {\bibinfo {volume} {16}},\ \bibinfo {pages} {168}
  (\bibinfo {year} {2017}{\natexlab{b}})}\BibitemShut {NoStop}%
\bibitem [{\citenamefont {Zhang}\ \emph
  {et~al.}(2017{\natexlab{c}})\citenamefont {Zhang}, \citenamefont {Weng},
  \citenamefont {Zhang}, \citenamefont {Li}, \citenamefont {Luo},\ and\
  \citenamefont {Tan}}]{Zhang17-2}%
  \BibitemOpen
  \bibfield  {author} {\bibinfo {author} {\bibfnamefont {X.}~\bibnamefont
  {Zhang}}, \bibinfo {author} {\bibfnamefont {J.}~\bibnamefont {Weng}},
  \bibinfo {author} {\bibfnamefont {Z.}~\bibnamefont {Zhang}}, \bibinfo
  {author} {\bibfnamefont {X.}~\bibnamefont {Li}}, \bibinfo {author}
  {\bibfnamefont {W.}~\bibnamefont {Luo}}, \ and\ \bibinfo {author}
  {\bibfnamefont {X.}~\bibnamefont {Tan}},\ }\href
  {https://arxiv.org/abs/1712.08830} {\bibfield  {journal} {\bibinfo  {journal}
  {arXiv:1712.08830 [quant-ph]}\ } (\bibinfo {year}
  {2017}{\natexlab{c}})}\BibitemShut {NoStop}%
\bibitem [{\citenamefont {Zhang}\ \emph
  {et~al.}(2017{\natexlab{d}})\citenamefont {Zhang}, \citenamefont {Weng},
  \citenamefont {Luo},\ and\ \citenamefont {Tan}}]{Zhang17-3}%
  \BibitemOpen
  \bibfield  {author} {\bibinfo {author} {\bibfnamefont {X.}~\bibnamefont
  {Zhang}}, \bibinfo {author} {\bibfnamefont {J.}~\bibnamefont {Weng}},
  \bibinfo {author} {\bibfnamefont {W.}~\bibnamefont {Luo}}, \ and\ \bibinfo
  {author} {\bibfnamefont {X.}~\bibnamefont {Tan}},\ }\href
  {https://arxiv.org/abs/1712.08970} {\bibfield  {journal} {\bibinfo  {journal}
  {arXiv:1712.08970 [quant-ph]}\ } (\bibinfo {year}
  {2017}{\natexlab{d}})}\BibitemShut {NoStop}%
\bibitem [{\citenamefont {Croke}\ and\ \citenamefont
  {Barnett}(2017)}]{Corke17}%
  \BibitemOpen
  \bibfield  {author} {\bibinfo {author} {\bibfnamefont {S.}~\bibnamefont
  {Croke}}\ and\ \bibinfo {author} {\bibfnamefont {S.~M.}\ \bibnamefont
  {Barnett}},\ }\href {\doibase 10.1103/PhysRevA.95.012337} {\bibfield
  {journal} {\bibinfo  {journal} {Phys. Rev. A}\ }\textbf {\bibinfo {volume}
  {95}},\ \bibinfo {pages} {012337} (\bibinfo {year} {2017})}\BibitemShut
  {NoStop}%
\bibitem [{\citenamefont {Halder}(2018)}]{Halder18}%
  \BibitemOpen
  \bibfield  {author} {\bibinfo {author} {\bibfnamefont {S.}~\bibnamefont
  {Halder}},\ }\href {\doibase 10.1103/PhysRevA.98.022303} {\bibfield
  {journal} {\bibinfo  {journal} {Phys. Rev. A}\ }\textbf {\bibinfo {volume}
  {98}},\ \bibinfo {pages} {022303} (\bibinfo {year} {2018})}\BibitemShut
  {NoStop}%
\bibitem [{\citenamefont {Halder}\ \emph
  {et~al.}(2019{\natexlab{a}})\citenamefont {Halder}, \citenamefont {Banik},\
  and\ \citenamefont {Ghosh}}]{Halder18-1}%
  \BibitemOpen
  \bibfield  {author} {\bibinfo {author} {\bibfnamefont {S.}~\bibnamefont
  {Halder}}, \bibinfo {author} {\bibfnamefont {M.}~\bibnamefont {Banik}}, \
  and\ \bibinfo {author} {\bibfnamefont {S.}~\bibnamefont {Ghosh}},\ }\href
  {\doibase 10.1103/PhysRevA.99.062329} {\bibfield  {journal} {\bibinfo
  {journal} {Phys. Rev. A}\ }\textbf {\bibinfo {volume} {99}},\ \bibinfo
  {pages} {062329} (\bibinfo {year} {2019}{\natexlab{a}})}\BibitemShut
  {NoStop}%
\bibitem [{\citenamefont {Halder}\ \emph
  {et~al.}(2019{\natexlab{b}})\citenamefont {Halder}, \citenamefont {Banik},
  \citenamefont {Agrawal},\ and\ \citenamefont {Bandyopadhyay}}]{Halder19}%
  \BibitemOpen
  \bibfield  {author} {\bibinfo {author} {\bibfnamefont {S.}~\bibnamefont
  {Halder}}, \bibinfo {author} {\bibfnamefont {M.}~\bibnamefont {Banik}},
  \bibinfo {author} {\bibfnamefont {S.}~\bibnamefont {Agrawal}}, \ and\
  \bibinfo {author} {\bibfnamefont {S.}~\bibnamefont {Bandyopadhyay}},\ }\href
  {\doibase 10.1103/PhysRevLett.122.040403} {\bibfield  {journal} {\bibinfo
  {journal} {Phys. Rev. Lett.}\ }\textbf {\bibinfo {volume} {122}},\ \bibinfo
  {pages} {040403} (\bibinfo {year} {2019}{\natexlab{b}})}\BibitemShut
  {NoStop}%
\bibitem [{\citenamefont {Rout}\ \emph {et~al.}(2019)\citenamefont {Rout},
  \citenamefont {Maity}, \citenamefont {Mukherjee}, \citenamefont {Halder},\
  and\ \citenamefont {Banik}}]{Rout19}%
  \BibitemOpen
  \bibfield  {author} {\bibinfo {author} {\bibfnamefont {S.}~\bibnamefont
  {Rout}}, \bibinfo {author} {\bibfnamefont {A.~G.}\ \bibnamefont {Maity}},
  \bibinfo {author} {\bibfnamefont {A.}~\bibnamefont {Mukherjee}}, \bibinfo
  {author} {\bibfnamefont {S.}~\bibnamefont {Halder}}, \ and\ \bibinfo {author}
  {\bibfnamefont {M.}~\bibnamefont {Banik}},\ }\href {\doibase
  10.1103/PhysRevA.100.032321} {\bibfield  {journal} {\bibinfo  {journal}
  {Phys. Rev. A}\ }\textbf {\bibinfo {volume} {100}},\ \bibinfo {pages}
  {032321} (\bibinfo {year} {2019})}\BibitemShut {NoStop}%
\bibitem [{\citenamefont {Cohen}(2007)}]{Cohen07}%
  \BibitemOpen
  \bibfield  {author} {\bibinfo {author} {\bibfnamefont {S.~M.}\ \bibnamefont
  {Cohen}},\ }\href {\doibase 10.1103/PhysRevA.75.052313} {\bibfield  {journal}
  {\bibinfo  {journal} {Phys. Rev. A}\ }\textbf {\bibinfo {volume} {75}},\
  \bibinfo {pages} {052313} (\bibinfo {year} {2007})}\BibitemShut {NoStop}%
\bibitem [{\citenamefont {Cohen}(2008)}]{Cohen08}%
  \BibitemOpen
  \bibfield  {author} {\bibinfo {author} {\bibfnamefont {S.~M.}\ \bibnamefont
  {Cohen}},\ }\href {\doibase 10.1103/PhysRevA.77.012304} {\bibfield  {journal}
  {\bibinfo  {journal} {Phys. Rev. A}\ }\textbf {\bibinfo {volume} {77}},\
  \bibinfo {pages} {012304} (\bibinfo {year} {2008})}\BibitemShut {NoStop}%
\bibitem [{\citenamefont {Bandyopadhyay}\ \emph {et~al.}(2009)\citenamefont
  {Bandyopadhyay}, \citenamefont {Brassard}, \citenamefont {Kimmel},\ and\
  \citenamefont {Wootters}}]{Bandyopadhyay09-1}%
  \BibitemOpen
  \bibfield  {author} {\bibinfo {author} {\bibfnamefont {S.}~\bibnamefont
  {Bandyopadhyay}}, \bibinfo {author} {\bibfnamefont {G.}~\bibnamefont
  {Brassard}}, \bibinfo {author} {\bibfnamefont {S.}~\bibnamefont {Kimmel}}, \
  and\ \bibinfo {author} {\bibfnamefont {W.~K.}\ \bibnamefont {Wootters}},\
  }\href {\doibase 10.1103/PhysRevA.80.012313} {\bibfield  {journal} {\bibinfo
  {journal} {Phys. Rev. A}\ }\textbf {\bibinfo {volume} {80}},\ \bibinfo
  {pages} {012313} (\bibinfo {year} {2009})}\BibitemShut {NoStop}%
\bibitem [{\citenamefont {Bandyopadhyay}\ \emph {et~al.}(2010)\citenamefont
  {Bandyopadhyay}, \citenamefont {Rahaman},\ and\ \citenamefont
  {Wootters}}]{Bandyopadhyay10}%
  \BibitemOpen
  \bibfield  {author} {\bibinfo {author} {\bibfnamefont {S.}~\bibnamefont
  {Bandyopadhyay}}, \bibinfo {author} {\bibfnamefont {R.}~\bibnamefont
  {Rahaman}}, \ and\ \bibinfo {author} {\bibfnamefont {W.~K.}\ \bibnamefont
  {Wootters}},\ }\href {\doibase 10.1088/1751-8113/43/45/455303} {\bibfield
  {journal} {\bibinfo  {journal} {J. Phys. A: Math. Theor.}\ }\textbf {\bibinfo
  {volume} {43}},\ \bibinfo {pages} {455303} (\bibinfo {year}
  {2010})}\BibitemShut {NoStop}%
\bibitem [{\citenamefont {{Yu}}\ \emph {et~al.}(2014)\citenamefont {{Yu}},
  \citenamefont {{Duan}},\ and\ \citenamefont {{Ying}}}]{Duan14}%
  \BibitemOpen
  \bibfield  {author} {\bibinfo {author} {\bibfnamefont {N.}~\bibnamefont
  {{Yu}}}, \bibinfo {author} {\bibfnamefont {R.}~\bibnamefont {{Duan}}}, \ and\
  \bibinfo {author} {\bibfnamefont {M.}~\bibnamefont {{Ying}}},\ }\href
  {\doibase 10.1109/TIT.2014.2307575} {\bibfield  {journal} {\bibinfo
  {journal} {IEEE Trans. Inf. Theory}\ }\textbf {\bibinfo {volume} {60}},\
  \bibinfo {pages} {2069} (\bibinfo {year} {2014})}\BibitemShut {NoStop}%
\bibitem [{\citenamefont {{Bandyopadhyay}}\ \emph {et~al.}(2015)\citenamefont
  {{Bandyopadhyay}}, \citenamefont {{Cosentino}}, \citenamefont {{Johnston}},
  \citenamefont {{Russo}}, \citenamefont {{Watrous}},\ and\ \citenamefont
  {{Yu}}}]{Bandyopadhyay14}%
  \BibitemOpen
  \bibfield  {author} {\bibinfo {author} {\bibfnamefont {S.}~\bibnamefont
  {{Bandyopadhyay}}}, \bibinfo {author} {\bibfnamefont {A.}~\bibnamefont
  {{Cosentino}}}, \bibinfo {author} {\bibfnamefont {N.}~\bibnamefont
  {{Johnston}}}, \bibinfo {author} {\bibfnamefont {V.}~\bibnamefont {{Russo}}},
  \bibinfo {author} {\bibfnamefont {J.}~\bibnamefont {{Watrous}}}, \ and\
  \bibinfo {author} {\bibfnamefont {N.}~\bibnamefont {{Yu}}},\ }\href {\doibase
  10.1109/TIT.2015.2417755} {\bibfield  {journal} {\bibinfo  {journal} {IEEE
  Trans. Inf. Theory}\ }\textbf {\bibinfo {volume} {61}},\ \bibinfo {pages}
  {3593} (\bibinfo {year} {2015})}\BibitemShut {NoStop}%
\bibitem [{\citenamefont {Bandyopadhyay}\ \emph {et~al.}(2016)\citenamefont
  {Bandyopadhyay}, \citenamefont {Halder},\ and\ \citenamefont
  {Nathanson}}]{Bandyopadhyay16}%
  \BibitemOpen
  \bibfield  {author} {\bibinfo {author} {\bibfnamefont {S.}~\bibnamefont
  {Bandyopadhyay}}, \bibinfo {author} {\bibfnamefont {S.}~\bibnamefont
  {Halder}}, \ and\ \bibinfo {author} {\bibfnamefont {M.}~\bibnamefont
  {Nathanson}},\ }\href {\doibase 10.1103/PhysRevA.94.022311} {\bibfield
  {journal} {\bibinfo  {journal} {Phys. Rev. A}\ }\textbf {\bibinfo {volume}
  {94}},\ \bibinfo {pages} {022311} (\bibinfo {year} {2016})}\BibitemShut
  {NoStop}%
\bibitem [{\citenamefont {Zhang}\ \emph
  {et~al.}(2016{\natexlab{c}})\citenamefont {Zhang}, \citenamefont {Gao},
  \citenamefont {Cao}, \citenamefont {Qin},\ and\ \citenamefont
  {Wen}}]{Zhang16-2}%
  \BibitemOpen
  \bibfield  {author} {\bibinfo {author} {\bibfnamefont {Z.-C.}\ \bibnamefont
  {Zhang}}, \bibinfo {author} {\bibfnamefont {F.}~\bibnamefont {Gao}}, \bibinfo
  {author} {\bibfnamefont {T.-Q.}\ \bibnamefont {Cao}}, \bibinfo {author}
  {\bibfnamefont {S.-J.}\ \bibnamefont {Qin}}, \ and\ \bibinfo {author}
  {\bibfnamefont {Q.-Y.}\ \bibnamefont {Wen}},\ }\href {\doibase
  https://dx.doi.org/10.1038/srep30493} {\bibfield  {journal} {\bibinfo
  {journal} {Sci. Rep.}\ }\textbf {\bibinfo {volume} {6}},\ \bibinfo {pages}
  {30493} (\bibinfo {year} {2016}{\natexlab{c}})}\BibitemShut {NoStop}%
\bibitem [{\citenamefont {Bandyopadhyay}\ \emph {et~al.}(2018)\citenamefont
  {Bandyopadhyay}, \citenamefont {Halder},\ and\ \citenamefont
  {Nathanson}}]{Bandyopadhyay18}%
  \BibitemOpen
  \bibfield  {author} {\bibinfo {author} {\bibfnamefont {S.}~\bibnamefont
  {Bandyopadhyay}}, \bibinfo {author} {\bibfnamefont {S.}~\bibnamefont
  {Halder}}, \ and\ \bibinfo {author} {\bibfnamefont {M.}~\bibnamefont
  {Nathanson}},\ }\href {\doibase 10.1103/PhysRevA.97.022314} {\bibfield
  {journal} {\bibinfo  {journal} {Phys. Rev. A}\ }\textbf {\bibinfo {volume}
  {97}},\ \bibinfo {pages} {022314} (\bibinfo {year} {2018})}\BibitemShut
  {NoStop}%
\bibitem [{\citenamefont {Zhang}\ \emph {et~al.}(2018)\citenamefont {Zhang},
  \citenamefont {Song}, \citenamefont {Song}, \citenamefont {Gao},
  \citenamefont {Qin},\ and\ \citenamefont {Wen}}]{Zhang18}%
  \BibitemOpen
  \bibfield  {author} {\bibinfo {author} {\bibfnamefont {Z.-C.}\ \bibnamefont
  {Zhang}}, \bibinfo {author} {\bibfnamefont {Y.-Q.}\ \bibnamefont {Song}},
  \bibinfo {author} {\bibfnamefont {T.-T.}\ \bibnamefont {Song}}, \bibinfo
  {author} {\bibfnamefont {F.}~\bibnamefont {Gao}}, \bibinfo {author}
  {\bibfnamefont {S.-J.}\ \bibnamefont {Qin}}, \ and\ \bibinfo {author}
  {\bibfnamefont {Q.-Y.}\ \bibnamefont {Wen}},\ }\href {\doibase
  10.1103/PhysRevA.97.022334} {\bibfield  {journal} {\bibinfo  {journal} {Phys.
  Rev. A}\ }\textbf {\bibinfo {volume} {97}},\ \bibinfo {pages} {022334}
  (\bibinfo {year} {2018})}\BibitemShut {NoStop}%
\bibitem [{\citenamefont {Li}\ \emph {et~al.}(2019)\citenamefont {Li},
  \citenamefont {Gao}, \citenamefont {Zhang},\ and\ \citenamefont
  {Wen}}]{Li19}%
  \BibitemOpen
  \bibfield  {author} {\bibinfo {author} {\bibfnamefont {L.-J.}\ \bibnamefont
  {Li}}, \bibinfo {author} {\bibfnamefont {F.}~\bibnamefont {Gao}}, \bibinfo
  {author} {\bibfnamefont {Z.-C.}\ \bibnamefont {Zhang}}, \ and\ \bibinfo
  {author} {\bibfnamefont {Q.-Y.}\ \bibnamefont {Wen}},\ }\href {\doibase
  10.1103/PhysRevA.99.012343} {\bibfield  {journal} {\bibinfo  {journal} {Phys.
  Rev. A}\ }\textbf {\bibinfo {volume} {99}},\ \bibinfo {pages} {012343}
  (\bibinfo {year} {2019})}\BibitemShut {NoStop}%
\bibitem [{\citenamefont {Bennett}\ \emph {et~al.}(1993)\citenamefont
  {Bennett}, \citenamefont {Brassard}, \citenamefont {Cr\'epeau}, \citenamefont
  {Jozsa}, \citenamefont {Peres},\ and\ \citenamefont {Wootters}}]{Bennett93}%
  \BibitemOpen
  \bibfield  {author} {\bibinfo {author} {\bibfnamefont {C.~H.}\ \bibnamefont
  {Bennett}}, \bibinfo {author} {\bibfnamefont {G.}~\bibnamefont {Brassard}},
  \bibinfo {author} {\bibfnamefont {C.}~\bibnamefont {Cr\'epeau}}, \bibinfo
  {author} {\bibfnamefont {R.}~\bibnamefont {Jozsa}}, \bibinfo {author}
  {\bibfnamefont {A.}~\bibnamefont {Peres}}, \ and\ \bibinfo {author}
  {\bibfnamefont {W.~K.}\ \bibnamefont {Wootters}},\ }\href {\doibase
  10.1103/PhysRevLett.70.1895} {\bibfield  {journal} {\bibinfo  {journal}
  {Phys. Rev. Lett.}\ }\textbf {\bibinfo {volume} {70}},\ \bibinfo {pages}
  {1895} (\bibinfo {year} {1993})}\BibitemShut {NoStop}%
\bibitem [{\citenamefont {Bennett}\ and\ \citenamefont
  {Wiesner}(1992)}]{Bennett92}%
  \BibitemOpen
  \bibfield  {author} {\bibinfo {author} {\bibfnamefont {C.~H.}\ \bibnamefont
  {Bennett}}\ and\ \bibinfo {author} {\bibfnamefont {S.~J.}\ \bibnamefont
  {Wiesner}},\ }\href {\doibase 10.1103/PhysRevLett.69.2881} {\bibfield
  {journal} {\bibinfo  {journal} {Phys. Rev. Lett.}\ }\textbf {\bibinfo
  {volume} {69}},\ \bibinfo {pages} {2881} (\bibinfo {year}
  {1992})}\BibitemShut {NoStop}%
\bibitem [{\citenamefont {Ekert}(1991)}]{Ekert91}%
  \BibitemOpen
  \bibfield  {author} {\bibinfo {author} {\bibfnamefont {A.~K.}\ \bibnamefont
  {Ekert}},\ }\href {\doibase 10.1103/PhysRevLett.67.661} {\bibfield  {journal}
  {\bibinfo  {journal} {Phys. Rev. Lett.}\ }\textbf {\bibinfo {volume} {67}},\
  \bibinfo {pages} {661} (\bibinfo {year} {1991})}\BibitemShut {NoStop}%
\bibitem [{\citenamefont {Cubitt}\ \emph {et~al.}(2003)\citenamefont {Cubitt},
  \citenamefont {Verstraete}, \citenamefont {D\"ur},\ and\ \citenamefont
  {Cirac}}]{Cubitt03}%
  \BibitemOpen
  \bibfield  {author} {\bibinfo {author} {\bibfnamefont {T.~S.}\ \bibnamefont
  {Cubitt}}, \bibinfo {author} {\bibfnamefont {F.}~\bibnamefont {Verstraete}},
  \bibinfo {author} {\bibfnamefont {W.}~\bibnamefont {D\"ur}}, \ and\ \bibinfo
  {author} {\bibfnamefont {J.~I.}\ \bibnamefont {Cirac}},\ }\href {\doibase
  10.1103/PhysRevLett.91.037902} {\bibfield  {journal} {\bibinfo  {journal}
  {Phys. Rev. Lett.}\ }\textbf {\bibinfo {volume} {91}},\ \bibinfo {pages}
  {037902} (\bibinfo {year} {2003})}\BibitemShut {NoStop}%
\bibitem [{\citenamefont {Mi\ifmmode~\check{s}\else \v{s}\fi{}ta}\ and\
  \citenamefont {Korolkova}(2008)}]{Mista08}%
  \BibitemOpen
  \bibfield  {author} {\bibinfo {author} {\bibfnamefont {L.}~\bibnamefont
  {Mi\ifmmode~\check{s}\else \v{s}\fi{}ta}}\ and\ \bibinfo {author}
  {\bibfnamefont {N.}~\bibnamefont {Korolkova}},\ }\href {\doibase
  10.1103/PhysRevA.77.050302} {\bibfield  {journal} {\bibinfo  {journal} {Phys.
  Rev. A}\ }\textbf {\bibinfo {volume} {77}},\ \bibinfo {pages} {050302}
  (\bibinfo {year} {2008})}\BibitemShut {NoStop}%
\bibitem [{\citenamefont {Mi\ifmmode~\check{s}\else \v{s}\fi{}ta}\ and\
  \citenamefont {Korolkova}(2009)}]{Mista09}%
  \BibitemOpen
  \bibfield  {author} {\bibinfo {author} {\bibfnamefont {L.}~\bibnamefont
  {Mi\ifmmode~\check{s}\else \v{s}\fi{}ta}}\ and\ \bibinfo {author}
  {\bibfnamefont {N.}~\bibnamefont {Korolkova}},\ }\href {\doibase
  10.1103/PhysRevA.80.032310} {\bibfield  {journal} {\bibinfo  {journal} {Phys.
  Rev. A}\ }\textbf {\bibinfo {volume} {80}},\ \bibinfo {pages} {032310}
  (\bibinfo {year} {2009})}\BibitemShut {NoStop}%
\bibitem [{\citenamefont {Mi\ifmmode~\check{s}\else
  \v{s}\fi{}ta}(2013)}]{Mista13}%
  \BibitemOpen
  \bibfield  {author} {\bibinfo {author} {\bibfnamefont {L.}~\bibnamefont
  {Mi\ifmmode~\check{s}\else \v{s}\fi{}ta}},\ }\href {\doibase
  10.1103/PhysRevA.87.062326} {\bibfield  {journal} {\bibinfo  {journal} {Phys.
  Rev. A}\ }\textbf {\bibinfo {volume} {87}},\ \bibinfo {pages} {062326}
  (\bibinfo {year} {2013})}\BibitemShut {NoStop}%
\bibitem [{\citenamefont {Kay}(2012)}]{Kay12}%
  \BibitemOpen
  \bibfield  {author} {\bibinfo {author} {\bibfnamefont {A.}~\bibnamefont
  {Kay}},\ }\href {\doibase 10.1103/PhysRevLett.109.080503} {\bibfield
  {journal} {\bibinfo  {journal} {Phys. Rev. Lett.}\ }\textbf {\bibinfo
  {volume} {109}},\ \bibinfo {pages} {080503} (\bibinfo {year}
  {2012})}\BibitemShut {NoStop}%
\bibitem [{\citenamefont {Fedrizzi}\ \emph {et~al.}(2013)\citenamefont
  {Fedrizzi}, \citenamefont {Zuppardo}, \citenamefont {Gillett}, \citenamefont
  {Broome}, \citenamefont {Almeida}, \citenamefont {Paternostro}, \citenamefont
  {White},\ and\ \citenamefont {Paterek}}]{Fedrizzi13}%
  \BibitemOpen
  \bibfield  {author} {\bibinfo {author} {\bibfnamefont {A.}~\bibnamefont
  {Fedrizzi}}, \bibinfo {author} {\bibfnamefont {M.}~\bibnamefont {Zuppardo}},
  \bibinfo {author} {\bibfnamefont {G.~G.}\ \bibnamefont {Gillett}}, \bibinfo
  {author} {\bibfnamefont {M.~A.}\ \bibnamefont {Broome}}, \bibinfo {author}
  {\bibfnamefont {M.~P.}\ \bibnamefont {Almeida}}, \bibinfo {author}
  {\bibfnamefont {M.}~\bibnamefont {Paternostro}}, \bibinfo {author}
  {\bibfnamefont {A.~G.}\ \bibnamefont {White}}, \ and\ \bibinfo {author}
  {\bibfnamefont {T.}~\bibnamefont {Paterek}},\ }\href {\doibase
  10.1103/PhysRevLett.111.230504} {\bibfield  {journal} {\bibinfo  {journal}
  {Phys. Rev. Lett.}\ }\textbf {\bibinfo {volume} {111}},\ \bibinfo {pages}
  {230504} (\bibinfo {year} {2013})}\BibitemShut {NoStop}%
\bibitem [{\citenamefont {Vollmer}\ \emph {et~al.}(2013)\citenamefont
  {Vollmer}, \citenamefont {Schulze}, \citenamefont {Eberle}, \citenamefont
  {H\"andchen}, \citenamefont {Fiur\'a\ifmmode~\check{s}\else \v{s}\fi{}ek},\
  and\ \citenamefont {Schnabel}}]{Vollmer13}%
  \BibitemOpen
  \bibfield  {author} {\bibinfo {author} {\bibfnamefont {C.~E.}\ \bibnamefont
  {Vollmer}}, \bibinfo {author} {\bibfnamefont {D.}~\bibnamefont {Schulze}},
  \bibinfo {author} {\bibfnamefont {T.}~\bibnamefont {Eberle}}, \bibinfo
  {author} {\bibfnamefont {V.}~\bibnamefont {H\"andchen}}, \bibinfo {author}
  {\bibfnamefont {J.}~\bibnamefont {Fiur\'a\ifmmode~\check{s}\else
  \v{s}\fi{}ek}}, \ and\ \bibinfo {author} {\bibfnamefont {R.}~\bibnamefont
  {Schnabel}},\ }\href {\doibase 10.1103/PhysRevLett.111.230505} {\bibfield
  {journal} {\bibinfo  {journal} {Phys. Rev. Lett.}\ }\textbf {\bibinfo
  {volume} {111}},\ \bibinfo {pages} {230505} (\bibinfo {year}
  {2013})}\BibitemShut {NoStop}%
\bibitem [{\citenamefont {Peuntinger}\ \emph {et~al.}(2013)\citenamefont
  {Peuntinger}, \citenamefont {Chille}, \citenamefont
  {Mi\ifmmode~\check{s}\else \v{s}\fi{}ta}, \citenamefont {Korolkova},
  \citenamefont {F\"ortsch}, \citenamefont {Korger}, \citenamefont
  {Marquardt},\ and\ \citenamefont {Leuchs}}]{Peuntinger13}%
  \BibitemOpen
  \bibfield  {author} {\bibinfo {author} {\bibfnamefont {C.}~\bibnamefont
  {Peuntinger}}, \bibinfo {author} {\bibfnamefont {V.}~\bibnamefont {Chille}},
  \bibinfo {author} {\bibfnamefont {L.}~\bibnamefont {Mi\ifmmode~\check{s}\else
  \v{s}\fi{}ta}}, \bibinfo {author} {\bibfnamefont {N.}~\bibnamefont
  {Korolkova}}, \bibinfo {author} {\bibfnamefont {M.}~\bibnamefont
  {F\"ortsch}}, \bibinfo {author} {\bibfnamefont {J.}~\bibnamefont {Korger}},
  \bibinfo {author} {\bibfnamefont {C.}~\bibnamefont {Marquardt}}, \ and\
  \bibinfo {author} {\bibfnamefont {G.}~\bibnamefont {Leuchs}},\ }\href
  {\doibase 10.1103/PhysRevLett.111.230506} {\bibfield  {journal} {\bibinfo
  {journal} {Phys. Rev. Lett.}\ }\textbf {\bibinfo {volume} {111}},\ \bibinfo
  {pages} {230506} (\bibinfo {year} {2013})}\BibitemShut {NoStop}%
\bibitem [{\citenamefont {Streltsov}\ \emph {et~al.}(2012)\citenamefont
  {Streltsov}, \citenamefont {Kampermann},\ and\ \citenamefont
  {Bru\ss{}}}]{Streltsov12}%
  \BibitemOpen
  \bibfield  {author} {\bibinfo {author} {\bibfnamefont {A.}~\bibnamefont
  {Streltsov}}, \bibinfo {author} {\bibfnamefont {H.}~\bibnamefont
  {Kampermann}}, \ and\ \bibinfo {author} {\bibfnamefont {D.}~\bibnamefont
  {Bru\ss{}}},\ }\href {\doibase 10.1103/PhysRevLett.108.250501} {\bibfield
  {journal} {\bibinfo  {journal} {Phys. Rev. Lett.}\ }\textbf {\bibinfo
  {volume} {108}},\ \bibinfo {pages} {250501} (\bibinfo {year}
  {2012})}\BibitemShut {NoStop}%
\bibitem [{\citenamefont {Chuan}\ \emph {et~al.}(2012)\citenamefont {Chuan},
  \citenamefont {Maillard}, \citenamefont {Modi}, \citenamefont {Paterek},
  \citenamefont {Paternostro},\ and\ \citenamefont {Piani}}]{Chuan12}%
  \BibitemOpen
  \bibfield  {author} {\bibinfo {author} {\bibfnamefont {T.~K.}\ \bibnamefont
  {Chuan}}, \bibinfo {author} {\bibfnamefont {J.}~\bibnamefont {Maillard}},
  \bibinfo {author} {\bibfnamefont {K.}~\bibnamefont {Modi}}, \bibinfo {author}
  {\bibfnamefont {T.}~\bibnamefont {Paterek}}, \bibinfo {author} {\bibfnamefont
  {M.}~\bibnamefont {Paternostro}}, \ and\ \bibinfo {author} {\bibfnamefont
  {M.}~\bibnamefont {Piani}},\ }\href {\doibase 10.1103/PhysRevLett.109.070501}
  {\bibfield  {journal} {\bibinfo  {journal} {Phys. Rev. Lett.}\ }\textbf
  {\bibinfo {volume} {109}},\ \bibinfo {pages} {070501} (\bibinfo {year}
  {2012})}\BibitemShut {NoStop}%
\bibitem [{\citenamefont {Streltsov}\ \emph {et~al.}(2014)\citenamefont
  {Streltsov}, \citenamefont {Kampermann},\ and\ \citenamefont
  {Bru\ss{}}}]{Streltsov14}%
  \BibitemOpen
  \bibfield  {author} {\bibinfo {author} {\bibfnamefont {A.}~\bibnamefont
  {Streltsov}}, \bibinfo {author} {\bibfnamefont {H.}~\bibnamefont
  {Kampermann}}, \ and\ \bibinfo {author} {\bibfnamefont {D.}~\bibnamefont
  {Bru\ss{}}},\ }\href {\doibase 10.1103/PhysRevA.90.032323} {\bibfield
  {journal} {\bibinfo  {journal} {Phys. Rev. A}\ }\textbf {\bibinfo {volume}
  {90}},\ \bibinfo {pages} {032323} (\bibinfo {year} {2014})}\BibitemShut
  {NoStop}%
\bibitem [{\citenamefont {Streltsov}\ \emph {et~al.}(2015)\citenamefont
  {Streltsov}, \citenamefont {Augusiak}, \citenamefont {Demianowicz},\ and\
  \citenamefont {Lewenstein}}]{Streltsov15}%
  \BibitemOpen
  \bibfield  {author} {\bibinfo {author} {\bibfnamefont {A.}~\bibnamefont
  {Streltsov}}, \bibinfo {author} {\bibfnamefont {R.}~\bibnamefont {Augusiak}},
  \bibinfo {author} {\bibfnamefont {M.}~\bibnamefont {Demianowicz}}, \ and\
  \bibinfo {author} {\bibfnamefont {M.}~\bibnamefont {Lewenstein}},\ }\href
  {\doibase 10.1103/PhysRevA.92.012335} {\bibfield  {journal} {\bibinfo
  {journal} {Phys. Rev. A}\ }\textbf {\bibinfo {volume} {92}},\ \bibinfo
  {pages} {012335} (\bibinfo {year} {2015})}\BibitemShut {NoStop}%
\bibitem [{\citenamefont {Zuppardo}\ \emph {et~al.}(2016)\citenamefont
  {Zuppardo}, \citenamefont {Krisnanda}, \citenamefont {Paterek}, \citenamefont
  {Bandyopadhyay}, \citenamefont {Banerjee}, \citenamefont {Deb}, \citenamefont
  {Halder}, \citenamefont {Modi},\ and\ \citenamefont
  {Paternostro}}]{Zuppardo16}%
  \BibitemOpen
  \bibfield  {author} {\bibinfo {author} {\bibfnamefont {M.}~\bibnamefont
  {Zuppardo}}, \bibinfo {author} {\bibfnamefont {T.}~\bibnamefont {Krisnanda}},
  \bibinfo {author} {\bibfnamefont {T.}~\bibnamefont {Paterek}}, \bibinfo
  {author} {\bibfnamefont {S.}~\bibnamefont {Bandyopadhyay}}, \bibinfo {author}
  {\bibfnamefont {A.}~\bibnamefont {Banerjee}}, \bibinfo {author}
  {\bibfnamefont {P.}~\bibnamefont {Deb}}, \bibinfo {author} {\bibfnamefont
  {S.}~\bibnamefont {Halder}}, \bibinfo {author} {\bibfnamefont
  {K.}~\bibnamefont {Modi}}, \ and\ \bibinfo {author} {\bibfnamefont
  {M.}~\bibnamefont {Paternostro}},\ }\href {\doibase
  10.1103/PhysRevA.93.012305} {\bibfield  {journal} {\bibinfo  {journal} {Phys.
  Rev. A}\ }\textbf {\bibinfo {volume} {93}},\ \bibinfo {pages} {012305}
  (\bibinfo {year} {2016})}\BibitemShut {NoStop}%
\bibitem [{\citenamefont {Horodecki}(1997)}]{Horodecki97}%
  \BibitemOpen
  \bibfield  {author} {\bibinfo {author} {\bibfnamefont {P.}~\bibnamefont
  {Horodecki}},\ }\href {\doibase
  https://doi.org/10.1016/S0375-9601(97)00416-7} {\bibfield  {journal}
  {\bibinfo  {journal} {Phys. Lett. A}\ }\textbf {\bibinfo {volume} {232}},\
  \bibinfo {pages} {333 } (\bibinfo {year} {1997})}\BibitemShut {NoStop}%
\bibitem [{\citenamefont {Horodecki}\ \emph {et~al.}(1998)\citenamefont
  {Horodecki}, \citenamefont {Horodecki},\ and\ \citenamefont
  {Horodecki}}]{Horodecki98}%
  \BibitemOpen
  \bibfield  {author} {\bibinfo {author} {\bibfnamefont {M.}~\bibnamefont
  {Horodecki}}, \bibinfo {author} {\bibfnamefont {P.}~\bibnamefont
  {Horodecki}}, \ and\ \bibinfo {author} {\bibfnamefont {R.}~\bibnamefont
  {Horodecki}},\ }\href {\doibase 10.1103/PhysRevLett.80.5239} {\bibfield
  {journal} {\bibinfo  {journal} {Phys. Rev. Lett.}\ }\textbf {\bibinfo
  {volume} {80}},\ \bibinfo {pages} {5239} (\bibinfo {year}
  {1998})}\BibitemShut {NoStop}%
\bibitem [{\citenamefont {Horodecki}\ \emph {et~al.}(2005)\citenamefont
  {Horodecki}, \citenamefont {Horodecki}, \citenamefont {Horodecki},\ and\
  \citenamefont {Oppenheim}}]{Horodecki05}%
  \BibitemOpen
  \bibfield  {author} {\bibinfo {author} {\bibfnamefont {K.}~\bibnamefont
  {Horodecki}}, \bibinfo {author} {\bibfnamefont {M.}~\bibnamefont
  {Horodecki}}, \bibinfo {author} {\bibfnamefont {P.}~\bibnamefont
  {Horodecki}}, \ and\ \bibinfo {author} {\bibfnamefont {J.}~\bibnamefont
  {Oppenheim}},\ }\href {\doibase 10.1103/PhysRevLett.94.160502} {\bibfield
  {journal} {\bibinfo  {journal} {Phys. Rev. Lett.}\ }\textbf {\bibinfo
  {volume} {94}},\ \bibinfo {pages} {160502} (\bibinfo {year}
  {2005})}\BibitemShut {NoStop}%
\bibitem [{\citenamefont {Horodecki}\ \emph {et~al.}(2008)\citenamefont
  {Horodecki}, \citenamefont {Pankowski}, \citenamefont {Horodecki},\ and\
  \citenamefont {Horodecki}}]{Horodecki08}%
  \BibitemOpen
  \bibfield  {author} {\bibinfo {author} {\bibfnamefont {K.}~\bibnamefont
  {Horodecki}}, \bibinfo {author} {\bibfnamefont {{\L}.}~\bibnamefont
  {Pankowski}}, \bibinfo {author} {\bibfnamefont {M.}~\bibnamefont
  {Horodecki}}, \ and\ \bibinfo {author} {\bibfnamefont {P.}~\bibnamefont
  {Horodecki}},\ }\href {\doibase 10.1109/TIT.2008.921709} {\bibfield
  {journal} {\bibinfo  {journal} {IEEE Trans. Inf. Theory}\ }\textbf {\bibinfo
  {volume} {54}},\ \bibinfo {pages} {2621} (\bibinfo {year}
  {2008})}\BibitemShut {NoStop}%
\bibitem [{\citenamefont {Horodecki}\ \emph {et~al.}(2009)\citenamefont
  {Horodecki}, \citenamefont {Horodecki}, \citenamefont {Horodecki},\ and\
  \citenamefont {Oppenheim}}]{Horodecki09}%
  \BibitemOpen
  \bibfield  {author} {\bibinfo {author} {\bibfnamefont {K.}~\bibnamefont
  {Horodecki}}, \bibinfo {author} {\bibfnamefont {M.}~\bibnamefont
  {Horodecki}}, \bibinfo {author} {\bibfnamefont {P.}~\bibnamefont
  {Horodecki}}, \ and\ \bibinfo {author} {\bibfnamefont {J.}~\bibnamefont
  {Oppenheim}},\ }\href {\doibase 10.1109/TIT.2008.2009798} {\bibfield
  {journal} {\bibinfo  {journal} {IEEE Trans. Inf. Theory}\ }\textbf {\bibinfo
  {volume} {55}},\ \bibinfo {pages} {1898} (\bibinfo {year}
  {2009})}\BibitemShut {NoStop}%
\bibitem [{\citenamefont {T\'{o}th}\ and\ \citenamefont
  {V\'{e}rtesi}(2018)}]{Toth18}%
  \BibitemOpen
  \bibfield  {author} {\bibinfo {author} {\bibfnamefont {G.}~\bibnamefont
  {T\'{o}th}}\ and\ \bibinfo {author} {\bibfnamefont {T.}~\bibnamefont
  {V\'{e}rtesi}},\ }\href {\doibase 10.1103/PhysRevLett.120.020506} {\bibfield
  {journal} {\bibinfo  {journal} {Phys. Rev. Lett.}\ }\textbf {\bibinfo
  {volume} {120}},\ \bibinfo {pages} {020506} (\bibinfo {year}
  {2018})}\BibitemShut {NoStop}%
\bibitem [{\citenamefont {Moroder}\ \emph {et~al.}(2014)\citenamefont
  {Moroder}, \citenamefont {Gittsovich}, \citenamefont {Huber},\ and\
  \citenamefont {G\"uhne}}]{Moroder14}%
  \BibitemOpen
  \bibfield  {author} {\bibinfo {author} {\bibfnamefont {T.}~\bibnamefont
  {Moroder}}, \bibinfo {author} {\bibfnamefont {O.}~\bibnamefont {Gittsovich}},
  \bibinfo {author} {\bibfnamefont {M.}~\bibnamefont {Huber}}, \ and\ \bibinfo
  {author} {\bibfnamefont {O.}~\bibnamefont {G\"uhne}},\ }\href {\doibase
  10.1103/PhysRevLett.113.050404} {\bibfield  {journal} {\bibinfo  {journal}
  {Phys. Rev. Lett.}\ }\textbf {\bibinfo {volume} {113}},\ \bibinfo {pages}
  {050404} (\bibinfo {year} {2014})}\BibitemShut {NoStop}%
\bibitem [{\citenamefont {V\'{e}rtesi}\ and\ \citenamefont
  {Brunner}(2014)}]{Vertesi14}%
  \BibitemOpen
  \bibfield  {author} {\bibinfo {author} {\bibfnamefont {T.}~\bibnamefont
  {V\'{e}rtesi}}\ and\ \bibinfo {author} {\bibfnamefont {N.}~\bibnamefont
  {Brunner}},\ }\href {\doibase 10.1038/ncomms6297} {\bibfield  {journal}
  {\bibinfo  {journal} {Nature Commun.}\ }\textbf {\bibinfo {volume} {5}},\
  \bibinfo {pages} {5297} (\bibinfo {year} {2014})}\BibitemShut {NoStop}%
\bibitem [{\citenamefont {Yu}\ and\ \citenamefont {Oh}(2017)}]{Yu17}%
  \BibitemOpen
  \bibfield  {author} {\bibinfo {author} {\bibfnamefont {S.}~\bibnamefont
  {Yu}}\ and\ \bibinfo {author} {\bibfnamefont {C.~H.}\ \bibnamefont {Oh}},\
  }\href {\doibase 10.1103/PhysRevA.95.032111} {\bibfield  {journal} {\bibinfo
  {journal} {Phys. Rev. A}\ }\textbf {\bibinfo {volume} {95}},\ \bibinfo
  {pages} {032111} (\bibinfo {year} {2017})}\BibitemShut {NoStop}%
\bibitem [{\citenamefont {Sengupta}\ and\ \citenamefont
  {Arvind}(2013)}]{Sengupta13}%
  \BibitemOpen
  \bibfield  {author} {\bibinfo {author} {\bibfnamefont {R.}~\bibnamefont
  {Sengupta}}\ and\ \bibinfo {author} {\bibnamefont {Arvind}},\ }\href
  {\doibase 10.1103/PhysRevA.87.012318} {\bibfield  {journal} {\bibinfo
  {journal} {Phys. Rev. A}\ }\textbf {\bibinfo {volume} {87}},\ \bibinfo
  {pages} {012318} (\bibinfo {year} {2013})}\BibitemShut {NoStop}%
\bibitem [{\citenamefont {Halder}\ and\ \citenamefont
  {Sengupta}(2019)}]{Halder19-1}%
  \BibitemOpen
  \bibfield  {author} {\bibinfo {author} {\bibfnamefont {S.}~\bibnamefont
  {Halder}}\ and\ \bibinfo {author} {\bibfnamefont {R.}~\bibnamefont
  {Sengupta}},\ }\href {\doibase
  https://doi.org/10.1016/j.physleta.2019.04.003} {\bibfield  {journal}
  {\bibinfo  {journal} {Phys. Lett. A}\ }\textbf {\bibinfo {volume} {383}},\
  \bibinfo {pages} {2004 } (\bibinfo {year} {2019})}\BibitemShut {NoStop}%
\bibitem [{\citenamefont {Choi}(1975)}]{Choi75}%
  \BibitemOpen
  \bibfield  {author} {\bibinfo {author} {\bibfnamefont {M.~D.}\ \bibnamefont
  {Choi}},\ }\href {\doibase 10.1016/0024-3795(75)90058-0} {\bibfield
  {journal} {\bibinfo  {journal} {Linear Algebra Appl.}\ }\textbf {\bibinfo
  {volume} {12}},\ \bibinfo {pages} {95} (\bibinfo {year} {1975})}\BibitemShut
  {NoStop}%
\end{thebibliography}%
\end{document}